\documentclass[amsmath,amssymb,aps,pra,superscriptaddress,twocolumn]{revtex4-2}
\usepackage[subpreambles=true]{standalone}
\usepackage{lmodern}
\usepackage{stmaryrd}
\usepackage{graphicx,color}
\usepackage{subfigure}
\usepackage{units}
\usepackage{verbatim}
\usepackage{dsfont}
\usepackage{bm}
\usepackage{mathrsfs}
\usepackage{hyperref}
\hypersetup{colorlinks=true, linkcolor=black, citecolor=black, urlcolor=black}
\usepackage{amsmath,amssymb,amsthm,mathrsfs,bbm,bm,mathtools}
\usepackage{comment}
\usepackage{enumerate}
\usepackage{braket}
\usepackage{todonotes}
\usepackage[titletoc,title]{appendix}
\usepackage{xypic}
\usepackage{multirow}
\usepackage{soul}
\usepackage{bibunits}

\newtheorem{definition}{Definition}

\newtheorem{theorem}{Theorem}

\newtheorem{lemma}[theorem]{Lemma}

\newtheorem{innercustomlemma}{Lemma}
\newenvironment{customlemma}[1]
  {\renewcommand\theinnercustomlemma{#1}\innercustomlemma}
  {\endinnercustomlemma}

\newtheorem{innercustomcor}{Corollary}
\newenvironment{customcor}[1]
  {\renewcommand\theinnercustomcor{#1}\innercustomcor}
  {\endinnercustomcor}

\begin{document}
\title{Erratum: Observing a changing Hilbert-space inner product}
\author{Salini Karuvade}
\affiliation{Centre for Engineered Quantum Systems, School of Physics, 
University of Sydney, Sydney, New South Wales 2006, Australia}
\author{Abhijeet Alase }
\affiliation{Centre for Engineered Quantum Systems, School of Physics, 
University of Sydney, Sydney, New South Wales 2006, Australia}

\author{Jacob L.\ Barnett}
\affiliation{Perimeter Institute for Theoretical Physics, 31 Caroline Street North, Waterloo, Ontario N2J 2Y5, Canada}
\affiliation{Department of Physics \& Astronomy, University of Waterloo, Waterloo, Ontario
N2L 3G1, Canada}
\author{Barry C.\ Sanders}
\affiliation{Institute for Quantum Science and Technology, University of Calgary, Alberta T2N 1N4, Canada}

\maketitle

\onecolumngrid
Lemma~2 
in Appendix~A of our original paper,
which was
first presented in 1992~\cite{SGH92},
is incorrect as we explain in this erratum.
This lemma should be replaced by the following lemma and corollary.
\begin{customlemma}{2}
A bounded positive-definite operator, $\eta\in\mathcal{B}\left(\mathscr{H}\right)$, is invertible if and only if $\eta$ is surjective. 
\end{customlemma} 
\begin{customcor}{2.1}\label{cor}
Let $\eta\in\mathcal{B}\left(\mathscr{H}\right)$ 
be a surjective, positive-definite operator.
Then, $\eta^{-1}$ is bounded and positive-definite.
\end{customcor} 

\noindent Neither paper articulated a complete set of constraints required for the invertibility of the metric 
operator $\eta$.
Specifically, the inverse of a bounded positive-definite operator, $\eta\in\mathcal{B}\left(\mathscr{H}\right)$,
in the absence of the surjectivity constraint
is not necessarily bounded. 

A simple counter-example  to the original Lemma 2 is now presented.
Consider a separable Hilbert space $\mathscr{H}$ with a countably infinite orthonormal basis, $\{\ket{e_j}, j \in \mathbb{N}\}$. Consider the bounded positive-definite operator 
\begin{equation}\label{eq:cex}
    \eta: \ket{e_j} \mapsto \frac{1}{j}\ket{e_j} \quad \forall j \in \mathbb{N}.
\end{equation}
On the dense subspace of $\mathscr{H}$ spanned by finite linear combinations of 
$\{\ket{e_j}, j \in \mathbb{N}\}$ the inverse of $\eta$ is given
by 
\begin{equation}
    \eta^{-1}:\ket{e_j} \mapsto j\ket{e_j} \quad \forall j \in \mathbb{N},
\end{equation}
which is clearly not bounded
on the whole $\mathscr{H}$.

For $\eta$ to be invertible, we need to enforce an additional constraint,
namely that the range of $\eta$ coincides with $\mathscr{H}$, or in other words, $\eta$ is surjective.
Below we prove the corrected version of Lemma~2, which is stated above.
\begin{proof}
An operator is injective if and only if its kernel contains only the zero vector. With this fact in mind, a straightforward proof by contradiction demonstrates the injectivity of all positive-definite operators (if  $\ket{\psi}\in\ker(\eta) \setminus \bm{0}$, then 
$\braket{\psi|\eta|\psi}=0$, thereby violating the positive-definiteness of $\eta$). Together with the supposition that $\eta$ is surjective, we find that $\eta$ is a bijection. Thus, the bounded inverse theorem \cite[Thm.~14.5.1]{Narici2010} implies $\eta^{-1}\in\mathcal{B}\left(\mathscr{H}\right)$.
\end{proof}

Corollary~\ref{cor} now follows from
$\braket{\psi|\eta^{-1}|\psi} = (\bra{\psi}\eta^{-1})\eta(\eta^{-1}\ket{\psi}) > 0$ 
for any $\ket{\psi} \in \mathscr{H}$,
where we used the positive-definiteness of $\eta$ in the last inequality.
We remark that self-adjointness of the metric operator follows from 
its positive-definiteness by polarization identity~\cite{BB03},
and therefore need not be assumed separately as in the original 
version of Lemma~2. 
Similarly, if $\eta\in \mathcal{B}(\mathscr{H})$ is positive-definite and surjective,
then $\eta^{-1} \in \mathcal{B}(\mathscr{H})$ is positive-definite
by Corollary~\ref{cor} above, and therefore self-adjoint.

Invertibility of the metric operator also guarantees that $\mathscr{H}_\eta$ defined in
the original paper is a Hilbert space with respect to the modified inner product 
$\braket{\bullet|\bullet}_\eta$~\cite[Appendix A]{SGH92}. 
Subsequently, surjectivity of $\eta$,
in addition to boundedness and positive-definiteness, is required for defining 
the change in representation  $\mathcal{R}_\eta$ [Eq.~(2) of the original paper] and 
the inner-product changing operation $\mathcal{E}_{\eta}$Eq.~(4) of the original paper].

Our results concerning the simulation of PT-symmetric Hamiltonians hold without any
additional assumptions.
This is because we exclusively work with unbroken PT-symmetric Hamiltonians in finite dimensions. The rank-nullity theorem implies that all injective operators with finite-dimensional domains are additionally surjective. Therefore, Lemma~2 as originally stated in the paper is valid
for any finite-dimensional Hilbert space $\mathscr{H}$. 

We also note that for any unbroken PT-symmetric Hamiltonian in 
a separable, infinite dimensional Hilbert space $\mathscr{H}$, an $\eta \in \mathcal{B}(\mathscr{H})$
satisfying the quasi-Hermiticity condition 
that is both surjective and positive-definite can be constructed~\cite[Thm.~4.3]{Kar22}
(see also~\cite[Eq.~(23)]{BiOrthogonal}).
This is a consequence of the fact that unbroken PT-symmetric Hamiltonians in infinite dimensions are defined 
to have eigenvectors forming a Riesz basis, in addition to these vectors 
being invariant under the action of the PT operator~\cite{Mos10b}.

In conclusion, amendment to Lemma~2 mandates that the operations $\mathcal{R}_\eta$ [Eq.~(2) of the original paper] and $\mathcal{E}_\eta$ [Eq.~(4) of the original paper] are to be defined only under the additional constraint that $\eta$ is surjective. 
The rest of our original paper is correct and remains unchanged after these modifications.

\ 

\paragraph*{Author Contribution Statement:} JLB pointed out the error in Lemma 2 of the original paper and provided the counter example given in this erratum. SK and AA developed the correct formulation of Lemma 2. SK and AA wrote the draft of the erratum with feedback from JLB and BCS.

\pagebreak

\begin{center}
{\bf \large Observing a changing Hilbert space inner product} 

\ 

{Salini Karuvade, Abhijeet Alase, and Barry C. Sanders}

{\small \it Institute for Quantum Science and Technology, University of Calgary,\\
2500 University Drive NW, Calgary, Alberta T2N 1N4, Canada}

\ 

\setlength{\fboxsep}{0pt}%
\setlength{\fboxrule}{0pt}%

\noindent\fbox{%
    \parbox{0.8\textwidth}{%
\small        
In quantum mechanics,
physical states are represented by rays in Hilbert space~$\mathscr H$,
which is a vector space imbued by an inner product~$\langle\,|\,\rangle$,
whose physical meaning arises as the overlap~$\langle\phi|\psi\rangle$
for~$\ket\psi$ a pure state (description of preparation)
and~$\bra\phi$ a projective measurement. However,
current quantum theory does not formally address the consequences of 
a changing inner product during the interval between preparation and measurement.
We establish a theoretical framework for such a changing inner product,
which we show is consistent with standard quantum mechanics.
Furthermore, we show that this change is described by a quantum operation,
which is tomographically observable,
and we elucidate how our result is strongly related to the 
exploding topic of PT-symmetric quantum mechanics.
We explain how to realize experimentally a changing inner product for a 
qubit in terms of a qutrit protocol with a unitary channel.
    }%
}

\end{center}

\twocolumngrid
Hilbert-space inner product is fundamental to quantum mechanics (QM), and its physicality relates 
to norm through the Born interpretation and to fidelity and distinguishability through its complex angle~\cite{SN21}.
The uniqueness of the inner product associated to a quantum system has come under scrutiny 
following the advent of PT-symmetric QM.
PT-symmetric systems are described by non-Hermitian Hamiltonians invariant under the combined action of parity~(P) and time~(T) 
inversion symmetries~\cite{BB98,BBJ02,BBJ03,Ben05}, and they are predicted to exhibit novel physical phenomena which
have been simulated on a variety of experimental platforms~\cite{RME+10,SLZ+11,BDG+12,POS+14,ZZS+16,XZB+17,EMK+18,WLG+19,ZLW+20}. 
These phenomena have been explained by observing that non-Hermitian Hamiltonians with 
unbroken PT symmetry are Hermitian with respect to a different Hilbert-space inner product~\cite{BBJ02,Mos03a,MECM08,JMCN19}.
Changing Hilbert-space inner-product is valuable for certain quantum information processing (QIP) tasks~\cite{Cro15} 
such as non-orthogonal state discrimination~\cite{BBC+13}, cloning~\cite{ZWX+20} and quantum algorithms~\cite{BBJM07,Mos09},
but perfunctory applications have led to counter-factual conclusions~\cite{Cro15,Pat14,CCC14} 
including violation of the no-signalling principle~\cite{YHFL14}. 
Our aim is to prescribe the correct procedure for changing Hilbert-space inner product and
to devise an experiment to validate our prescription.

Consistency  of a changing Hilbert-space inner product
with standard QM  and the 
unobservability of such a change in closed systems 
have been investigated.
A C$^*$-algebraic approach shows that
a set of non-Hermitian operators comprises the
observables of a
quantum mechanical system if and only if the operators are Hermitian with respect to
a new Hilbert-space inner product~\cite{SGH92}.
Such a modified inner product is the key to proving the equivalence
of PT-symmetric QM with 
the Dirac-von Neumann formulation of QM in the case of closed systems i.e., systems in which 
every time evolution is a unitary operation~\cite{BBJ02,Mos03a,Mos10a,Mos10b,Zno15,Mos18,ZWG19,JMCN19}.
Furthermore, 
this equivalence implies that any change in inner product is unobservable in experiments on closed systems~\cite{Bro16}.
Therefore, the above proposals that use the inner-product change for QIP tasks
as well as the counter-factual claims are not applicable to closed systems.

Evolution generated by PT-symmetric Hamiltonians has been implemented experimentally
for applications including sensing~\cite{LZO+16,COZ+17,HHW+17}, cloaking~\cite{ZFZ+13,SFA15}
and unidirectional propagation~\cite{RKEC10}.
These experiments simulate PT-symmetric dynamics on classical~\cite{RME+10,SLZ+11,POS+14,ZZS+16} or
quantum~\cite{BDG+12,XZB+17,LHL+19,XQW+19} systems by balancing loss and gain.
Another way to simulate  PT-symmetric Hamiltonians with real spectra is by 
dilating the non-unitary propagator to a non-local unitary operator over multiple subsystems, which has been demonstrated 
on qubit systems~\cite{GS08a,GS08b,ZHL13,TWY+16,KAU17,HKW18,WLG+19,XWZ+19,GZL+21}.
However, none of these simulation strategies involve effecting a change 
of inner product.

PT-symmetric Hamiltonians and a changing Hilbert-space inner product
are known to be consistent with standard QM for closed systems,
but they are not yet known to be consistent for open systems.
To solve these outstanding problems, we 
construct an operational framework, consistent with the C$^*$-algebraic formulation of QM,
which accommodates a change in inner product between preparation and measurement.
Furthermore, neither PT symmetry nor a changing Hilbert-space inner product are observable in
closed systems, but could be observable in open systems~\cite{Bro16}.
We show our change in inner product is implemented by a quantum operation
(henceforth assumed to be completely positive and trace non-increasing),
which can be observed using tomography.
Next we connect our framework to the burgeoning topic of PT-symmetric QM
by explaining how an inner-product-changing quantum operation can be 
used to implement PT-symmetric dynamics in an open system.
Finally, at the empirical level, we describe a potential experimental simulation for changing the inner product 
of a qubit by subjecting a qutrit to unitary evolution and 
neglecting the third Hilbert-space dimension during preparation and measurement but not during evolution.
We also extend this simulation procedure to $d$-dimensional systems.

To construct the operational framework for changing the inner product associated to a quantum system
between preparation and measurement,
we adopt the C$^*$-algebraic framework of QM~\cite{Str08},
which provides freedom in representing 
a given system on different Hilbert spaces following the Gel'fand-Naimark-Segal (GNS) construction~\cite{GN43,Seg47}.  
We employ this representation freedom 
first to construct representations of the C$^*$ algebra on a pair of Hilbert spaces 
whose inner products are related by a given metric operator~$\eta$. We then
define the change in inner product by~$\eta$ as the identity
isomorphism between the two Hilbert spaces.
To operationalize the change in inner product, we use commutative diagrams that connect 
this isomorphism to a quantum operation 
between the bounded operators on the two Hilbert spaces and 
finally observe that the quantum operation induces an observable physical transformation on the system.

In the operational approach, the operators of a quantum system form a unital C${}^*$ algebra
$\mathcal{A}=\{A\}$, which is equipped with a ${}^*$ operation that 
captures the notion of adjoint. The algebra~$\mathcal{A}$ is representable on a 
possibly infinite dimensional 
Hilbert space~$\mathscr H=(\mathscr{V},\langle\,\vert\,\rangle)$, 
comprising a complete vector space~$\mathscr V$
and an inner product~$\langle\,\vert\,\rangle$, which is a non-degenerate sesquilinear form.
In Fig.~\ref{fig:commutativediagram}(a),
observables are self-adjoint elements of~$\mathcal{A}$ and correspond to allowed measurements.
A representation of $\mathcal{A}$ is a product-preserving linear map 
\begin{equation}
\label{eq:pi*dagger}
\mathcal{A} \stackrel{\pi}{\to} \mathcal{B}(\mathscr{H}):
\pi(A^*) = (\pi(A))^\dagger,
\end{equation}
where~$\mathcal{B}(\mathscr H)$ denotes the space of bounded linear operators
acting on~$\mathscr H$ and $^\dagger$ denotes the Hermitian conjugate.
Such a representation can be obtained using the GNS construction~\cite{GN43,Seg47}.
Product preservation ensures that if $I$ is the identity operator in $\mathcal{A}$, then
$\pi(I)$ is the identity operator in~$\mathcal{B}(\mathscr H)$.
An operator $M\in \mathcal{B}(\mathscr H)$ satisfying $M^\dagger = M$ is called a self-adjoint or a Hermitian operator.
\begin{figure}
\begin{center}
\includegraphics[width = \columnwidth]{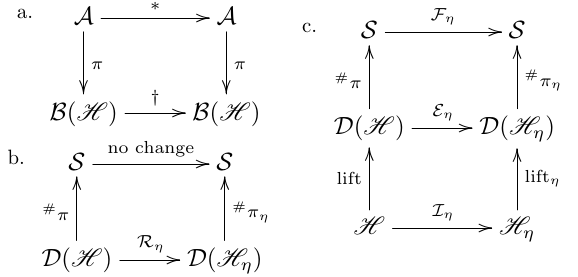}
\end{center}
\caption{
(a)~Diagram illustrating the relation between the ${}^{*}$ operation of $\mathcal{A}$ and the 
$\dagger$ operation of $\mathcal{B}(\mathscr{H})$ under the representation $\pi$.
(b) Commutative diagram depicting the change of representation from $\pi$ to $\pi_{\eta}$ under $\mathcal{R}_\eta$. Operationally, 
$\mathcal{R}_\eta$ represents a trivial transformation, i.e., no change, in $S$. (c)~Commutative diagram illustrating the relation between 
the maps $\mathcal{F}_\eta$,
$\mathcal{E}_\eta$ and $\mathcal{I}_{\eta}$. }
\label{fig:commutativediagram}
\end{figure}

We now define states and explain how to represent states as operators on Hilbert space.
States correspond to allowed preparations of the system
(Fig.~\ref{fig:commutativediagram}(b)).
A state $\omega$ is a positive linear functional on $\mathcal{A}$
that is normalized, i.e. $\omega(I) = 1$.
This definition is extended to include any subnormalized positive linear functional,
i.e.\ $\omega(I)\le 1$, which corresponds to probabilistic preparation in the state 
$\omega/\omega(I)$ with probability $\omega(I)$~\cite{Gud79}.
Supernormalized positive linear functionals are not valid states
according to this probabilistic interpretation~\footnote{
We remark that supernormalized functionals are valid states in some other frameworks
for describing system evolution under non-Hermitian Hamiltonians~\cite{Mos07,GKN10,UGRM12}
}.
Let $\mathcal{D}(\mathscr{H}):= \{\rho:\rho \ge 0,\rho=\rho^\dagger,
\text{tr}(\rho)\leq 1\} \subset \mathcal{B}(\mathscr{H})$ denote the set of density operators
with $\rho\in \mathcal{D}(\mathscr H)$ representing a state~$\omega$ 
by~$\rho\stackrel{\!\!{}^{\#}\!\!\pi}{\mapsto}\omega$
if and only if the expectation value~tr$(\rho\pi(A)) = \omega(A)\,\forall A$.
As ${}^\#\pi$ is uniquely determined by $\pi$, we say
that $\omega$ is represented by $\rho\in \mathcal{D}(\mathscr H)$ under $\pi$. 
We denote by  $\mathcal{S} = \{\omega\}$, the set of all states that are represented by density operators under $\pi$.
For the special case of pure~$\omega$,
$\rho = \ket{\psi}\bra{\psi}$ for some $\ket{\psi}\in \mathscr H$ with $\braket{\psi | \psi} \le 1$.
The transformation $\ket{\psi}\stackrel{\operatorname{lift}}{\mapsto}\ket{\psi}\bra{\psi}$ relates Hilbert-space
vectors to the density operators in~$\mathcal{D}(\mathscr H)$~\footnote{
The map lift is defined to act only on the normalized and subnormalized vectors in $\mathscr{H}$.
The domain of lift in Fig.~\ref{fig:commutativediagram}(c) is shown to be $\mathscr{H}$ for simplicity, but is specified 
rigorously in Appendix~\ref{sec:liftmap}.
}.

Now that we have explained states and their representations,
we now discuss changing representation
to being over a different Hilbert space.
Given a self-adjoint positive-definite metric operator~$\eta\in\mathcal{B}(\mathscr{H})$,
a new Hilbert space~$\mathscr{H}_\eta = (\mathscr{V},\langle\,\vert\,\rangle_\eta)$ can be constructed~\cite{SGH92} 
such that the inner products of the two Hilbert spaces are related by
$\langle \bullet\,\vert\,\bullet\rangle_\eta := \langle\bullet\vert\,\eta\bullet\rangle$.
A representation $\pi_{\eta}$ on $\mathscr{H}_{\eta}$ can be constructed 
through~(see Appendix~\ref{sec:newHilbertspace})
\begin{equation}
\pi(\bullet)\stackrel{\mathcal{R}_{\eta}}{\mapsto}\pi_{\eta}(\bullet):=R\pi(\bullet)R^{-1},
\mathcal{R}_\eta:\mathcal{B}(\mathscr H)\to\mathcal{B}(\mathscr H_\eta),
\end{equation}
where $R = \eta^{\nicefrac{-1}{2}}$ is a linear isometry from $\mathscr H$ to $\mathscr H_\eta$.
We note that this isometry has been used to prove that Hamiltonians with unbroken PT symmetry 
are consistent with standard QM, in the case of closed systems~\cite{GS08a,Mos10a,Mos10b,Zno15,Cro15}.
We refer to the quantum channel $\mathcal{R}_{\eta}$ as a `change in representation' (Fig.~\ref{fig:commutativediagram}b).
In representation~$\pi_{\eta}$, the state $\omega\stackrel{\;{}^{\#}\!\!\pi_{\eta}}{\mapsfrom}\mathcal{R}_{\eta}(\rho)$ 
such that tr$(\mathcal{R}_{\eta}(\rho)\pi_{\eta}(A)) = \omega(A)\forall A$,
and $\mathcal{R}_{\eta}(\rho) \ne \rho$ in general. 
For any pure state~$\omega$, $\exists\ket\psi \in \mathscr{H}_{\eta}$ such that $\omega\stackrel{\;{}^{\#}\!\!\pi_{\eta}}{\mapsfrom}\ket{\psi}\bra{\psi}
\eta\stackrel{\operatorname{lift}_{\eta}}{\mapsfrom}\ket{\psi}$~(see Appendix~\ref{sec:liftmap}).
As 
\begin{equation}
\label{eq:differentreps}
    \text{tr}(\rho \pi(A)) = \text{tr}(\mathcal{R}_{\eta}(\rho) \pi_\eta(A)) \;\;\forall \omega,A,
\end{equation}
representations $\pi$ and~$\pi_{\eta}$ are physically,
i.e. observationally, indistinguishable. 
The right-hand side of Eq.~\eqref{eq:differentreps} 
can also be interpreted as preparation (state) described in $\pi$
followed by a change in representation from~$\pi$ to~$\pi_\eta$ 
effected by $\mathcal{R}_{\eta}$ and finally measurement (observable) described in~$\pi_\eta$.
Change in representation between preparation and measurement sets the stage for 
our definition of change in inner product. 

We define a change in inner product by $\eta$
to be the identity isomorphism~$\mathcal{I}_{\eta}:\mathscr H \to \mathscr{H}_{\eta}$ such that
every $\ket{\psi} \mapsto \ket{\psi}$.
For any pair $\ket{\psi},\ket{\phi}\in\mathscr{H}$,
the inner product between the pair of transformed vectors $\mathcal{I}_{\eta}\ket{\psi},\ \mathcal{I}_{\eta}\ket{\phi}$ 
is $\braket{\psi|\eta|\phi}$,
and the change is trivial if $\eta = \pi(I)$; i.e.\ 
for all pairs $\ket{\psi},\ket{\phi}$, $\braket{\psi|\eta|\phi} = \braket{\psi|\phi}$.
Our definition is motivated by proposals to effect PT-symmetric evolution and measurement by 
changing the Hilbert-space inner product~\cite{BBJM07,BBC+13} but without a 
prescription for making such changes operationally or mathematically.
Next we explain separately, for the cases $\eta \le \pi(I)$ and $\eta \nleq \pi(I)$,
how the isomorphism $\mathcal{I}_\eta$ can be physically realized as a quantum operation.

The change in inner product by $\eta \le \pi(I)$ is physically realizable 
via the operation 
\begin{equation}\label{eq:Etilde}
    \mathcal{E}_\eta:\mathcal{B}\left(\mathscr H\right) \to \mathcal{B}\left(\mathscr{H}_{\eta}\right):M \mapsto M\eta,
\end{equation}
which is not trace-preserving for $\eta \ne \pi(I)$~(see Appendix~\ref{sec:channelEeta}).
The operation  $\mathcal{E}_\eta$ mimics the action of $\mathcal{I}_{\eta}$ at the level of density operators,
because $\ket{\psi} \stackrel{\operatorname{lift}}{\mapsto}\ket{\psi}\bra{\psi}$
whereas~$\mathcal{I}_{\eta}\ket{\psi} \stackrel{\operatorname{lift}_{\eta}}{\mapsto}\ket{\psi}\bra{\psi}\eta = 
\mathcal{E}_\eta(\ket{\psi}\bra{\psi})$.
The operation $\mathcal{E}_\eta$ induces a linear map 
$\mathcal{F}_\eta:\mathcal{S} \to \mathcal{S}$
such that, for any pure $\omega\in \mathcal{S}$, both $\omega$ and $\mathcal{F}_\eta(\omega)$ are represented by the 
same $\ket{\psi}$ under the representations $\pi$ and $\pi_{\eta}$ respectively.
However, $\omega$ and $\mathcal{F}_\eta(\omega)$ are not necessarily the same state (Fig.~\ref{fig:commutativediagram}c).
Even in the special case where the two states differ by a scaling factor,
they are inequivalent in our setting.
The expectation value of $I$ with respect to 
$\mathcal{F}_\eta(\omega)$ gives the success probability
of the inner-product changing quantum operation on the state $\omega$.
$\mathcal{E}_\eta$ can be implemented experimentally 
by lossy purity-preserving operations,
i.e., operations that are not necessarily deterministic and 
transform the set of pure states into itself.
In the Heisenberg picture,
the operators transform according to the map
\begin{equation}\label{eq:Edual}
\mathcal{E}^{\rm op}_\eta:\mathcal{B}\left(\mathscr H\right) \to \mathcal{B}\left(\mathscr{H}_{\eta}\right):M \mapsto \eta M.
\end{equation}
This transformation $\mathcal{E}^{\rm op}_\eta$ could modify  
commutator relations as we show in Appendix~\ref{sec:commutation}.

In the case $\eta \nleq \pi(I)$,
$\mathcal{E}_\eta$
is completely positive
but trace-increasing for some $\rho\in\mathcal{B}(\mathscr{H})$ 
and hence not a quantum operation.
In such cases, a scaled version of change in inner product can be implemented in the following way:
choose $\kappa\in(0,1)$ such that $\kappa\eta\leq \pi(I)$ and observe that 
$\mathcal{E}_{\kappa\eta} = \kappa\mathcal{E}_\eta$ with 
$\mathcal{E}_{\kappa\eta}$ a quantum operation. 
Therefore, $\mathcal{E}_{\kappa\eta}$  implements change in inner product by $\eta \nleq \pi(I)$
up to a scaling factor $\kappa$.
Such a scaled version of change in inner product is useful to reverse the effect of operation 
$\mathcal{E}_\eta$ when $\eta \leq \pi(I)$.
In this case, the isomorphism $\mathcal{I}_{\eta^{-1}}:\mathscr{H}_\eta\to\mathscr{H}$
reverses the change in inner product and the corresponding $\mathcal{E}_{\eta^{-1}}$
is not a valid operation because  $\eta^{-1} \ge \pi_{\eta}(I)$.
Nevertheless, we can choose $\kappa = \nicefrac{1}{\|\eta^{-1}\|}$, 
where $\|\bullet\|$ denotes the operator norm~\cite{Con07}, 
and observe that $\mathcal{E}_{\kappa\eta^{-1}}\circ\mathcal{E}_{\eta}(\rho)=\kappa \rho$ 
for all $\rho\in \mathcal{B}(\mathscr{H})$.
Therefore, the operation   
$\mathcal{E}_{\kappa\eta^{-1}}:\mathcal{B}(\mathscr{H}_{\eta})\to\mathcal{B}(\mathscr{H})$ 
reverses, with probability $\kappa$, the change in inner product by $\eta$.

The metric operator $\eta$ can be estimated via quantum process tomography~\cite{CN97} 
for $\eta\leq \pi(I)$,
or $\kappa \eta$ if otherwise.
The change in inner product by $\eta\leq \pi(I)$ is implemented via the operation 
(Fig.~\ref{fig:Etilde})
\begin{equation}
\label{eq:F}
\mathcal{E}_\eta= \mathcal{R}_{\eta}\circ \mathcal{G}_\eta,\,    
\mathcal{G}_\eta:\mathcal{B}(\mathscr H)\to \mathcal{B}(\mathscr H):M \mapsto \eta^{\nicefrac{1}{2}}M\eta^{\nicefrac{1}{2}},
\end{equation}
for the Kraus rank-1 operation $\mathcal{G}_\eta$.
Then the Kraus operator $\eta^{\nicefrac{1}{2}}$ and therefore $\eta$
can be estimated by quantum process tomography for trace non-increasing channels~\cite{BSS+10}.
In the other case $\eta\nleq \pi(I)$,
the change in inner product is implemented by the operation  
$\mathcal{E}_{\kappa\eta}$ from which $\kappa \eta$ is estimated similarly; however, 
the above procedure does not yield $\kappa$ and $\eta$ separately.
\begin{figure}
    \begin{center}
    \includegraphics[width = 0.7\columnwidth]{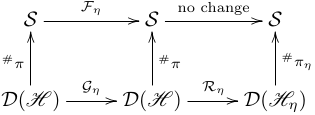}
    \end{center}
    \caption{Commutative diagram showing the action of $\mathcal{F}_\eta$ decomposed in terms of $\mathcal{G}_\eta$ and $\mathcal{R}_\eta$.}
    \label{fig:Etilde}
\end{figure}

We now discuss how to implement dynamics generated by 
a diagonalizable Hamiltonian $H_{\rm PT}$ with unbroken PT symmetry in finite dimensions
over a time $t \ge 0$,
by building on our framework for changing inner product.
The dynamical transformation generated by~$H_{\rm PT}$ is
\begin{equation}
\label{eq:UPT}
    \rho \stackrel{\mathcal{U}_{\rm PT}}{\mapsto} \kappa U_{\rm{PT}}\rho U^\dagger_{\rm{PT}}, \; U_{\rm{PT}} := \text{e}^{-\text{i}H_{\rm {PT}}t/\hbar},
\end{equation}
for some $\kappa\in(0,1)$, where both $\rho$, $\mathcal{U}_{\rm PT}(\rho)\in\mathcal{D}(\mathscr{H})$ represent states under $\pi$.
We show that this dynamics can be implemented by first changing the inner product, 
then applying an appropriate unitary channel and finally reversing the change in inner product.
To explain this sequence, we consider an arbitrary pure state represented by $\ket{\psi}\in \mathscr{H}$,
which is to be transformed to that represented by $\sqrt{\kappa}U_{\rm PT}\ket{\psi}\in \mathscr{H}$ 
(lower row of Fig.~\ref{fig:PTdynamics}).
Next, we compute a metric operator $\eta\leq \pi(I)\in \mathcal{B}(\mathscr H)$
that satisfies the quasi-Hermiticity condition
\begin{equation}\label{eq:quasiHermitian}
H^\dagger_{\rm {PT}} = \eta H_{\rm {PT}} \eta^{-1};
\end{equation}
the existence of such an $\eta$ is guaranteed 
as $H_{\rm PT}$ has unbroken PT symmetry~\cite{Mos02}.
The Hamiltonian $H_{\rm PT}$ is self-adjoint with respect to the inner product of the new Hilbert space $\mathscr{H}_{\eta}$.
Therefore, $U_{\rm PT}$ represents unitary dynamics on $\mathscr{H}_{\eta}$,
which constitutes the second step of the sequence.
Prior to implementing~$U_{\rm PT}$, we transform $\ket{\psi}\in\mathscr{H}$ to $\ket{\psi}\in\mathscr{H}_{\eta}$ 
via a change in inner product using $\mathcal{I}_\eta$.
Finally, the transformation from $U_{\rm PT}\ket{\psi}\in\mathscr{H}_{\eta}$ to $\sqrt{\kappa}U_{\rm PT}\ket{\psi}\in\mathscr{H}$
is equivalent to reversing the change in inner product using  $\mathcal{I}_{\eta^{-1}}$ with probability $\kappa$. 
This sequence extends to general mixed states by the application of lift, lift${}_\eta$ maps 
and linearity (upper row of Fig.~\ref{fig:PTdynamics});
here  $ \widetilde{\mathcal{U}}_{\rm{PT}}\in\mathcal{B}\left(\mathscr{H}_\eta\right)$ 
is the unitary channel satisfying
$\operatorname{lift}_\eta\left(U_{\rm PT}\ket{\psi}\right) =  \widetilde{\mathcal{U}}_{\rm{PT}}\left(\operatorname{lift}_\eta\left(\ket{\psi}\right)\right)$,
for all $\ket{\psi}\in\mathscr{H}_\eta$.

\begin{figure}
\begin{center}
\includegraphics[width = \columnwidth]{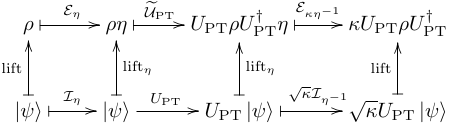}
\end{center}
\caption{Diagram showing implementation of a PT-symmetric dynamics using change in inner product and unitary dynamics.}
\label{fig:PTdynamics}
\end{figure}

PT-symmetric dynamics in Eq.~\eqref{eq:UPT} can be expressed as a sequence of channels 
acting exclusively on $\mathcal{B}(\mathscr{H})$, thereby paving the way for experimental simulation 
of PT-symmetric systems.
Following the upper row of Fig.~\ref{fig:PTdynamics}, 
we start by expressing $\mathcal{U}_{\rm{PT}}$ (Eq.~\eqref{eq:UPT})
as  $\mathcal{U}_{\rm{PT}}= \mathcal{E}_{\kappa \eta^{-1}} \circ  \widetilde{\mathcal{U}}_{\rm{PT}} \circ \mathcal{E}_\eta$.
Similar to~Eq.\ \eqref{eq:F},
we express the reverse change in inner product as 
$\mathcal{E}_{\kappa\eta^{-1}}= \mathcal{G}_{\kappa\eta^{-1}}\circ \mathcal{R}_{\kappa\eta^{-1}}$,
for  $\mathcal{G}_{\kappa\eta^{-1}}: \mathcal{B}(\mathscr{H})\to \mathcal{B}(\mathscr{H}): 
\mathcal{G}_{\kappa\eta^{-1}}(M) = \kappa \eta^{-\nicefrac{1}{2}}M\eta^{\nicefrac{1}{2}} $ and
$\mathcal{R}_{\kappa\eta^{-1}}: \mathcal{B}(\mathscr{H}_{\eta})\to \mathcal{B}(\mathscr{H})$ is the 
channel effecting the change in representation form $\pi_\eta$ to $\pi$.
We then rewrite~$\mathcal{U}_{\rm{PT}}$ as
\begin{equation}\label{eq:PTUnitaryoperational}
    \mathcal{U}_{\rm{PT}}= {\mathcal{G}}_{\kappa \eta^{-1}} \circ (\mathcal{R}_{\kappa\eta^{-1}} \circ \widetilde{\mathcal{U}}_{\rm{PT}} \circ \mathcal{R}_{\eta})\circ \mathcal{G}_\eta,
\end{equation}
which is the desired decomposition. The channels ${\mathcal{G}}_{\kappa \eta^{-1}}$, $ \mathcal{G}_\eta$ have 
single Kraus operators $\sqrt{\kappa}\eta^{-\nicefrac{1}{2}}$ and $\eta^{\nicefrac{1}{2}}$ respectively.
The maps $\mathcal{R}_{\eta}$,  $\mathcal{R}_{\kappa\eta^{-1}}$
only effect change in representation and operationally are equivalent to no change. 
Finally, the transformation $\mathcal{R}_{\kappa\eta^{-1}} \circ \widetilde{\mathcal{U}}_{\rm{PT}} \circ \mathcal{R}_{\eta}$
implements a channel 
corresponding to the unitary Kraus-operator $\eta^{\nicefrac{1}{2}}U_{\rm PT}\eta^{\nicefrac{-1}{2}}$ acting
on~$\mathscr{H}$, generated by the Hamiltonian 
\begin{equation}\label{eq:selfadjointHPT}
    h_{\rm PT} = \eta^{\nicefrac{1}{2}}H_{\rm PT}\eta^{\nicefrac{-1}{2}} \in \mathcal{B}(\mathscr{H}),
\end{equation}
which can be verified to be self-adjoint, i.e.\ $h_{\rm PT}^\dagger = h_{\rm PT}$, 
using the quasi-Hermiticity condition in Eq.~\eqref{eq:quasiHermitian}.

We now design a qutrit procedure for an agent to simulate successfully the change in inner product by $\eta \le \pi(I)$ 
of a qubit system with algebra $\mathcal{A}$, which is represented on a 
two-dimensional Hilbert space $\mathscr{H}_2$ by $\pi$.
Our procedure, which shall simulate the operation $\mathcal{G}_\eta$ (Fig.~\ref{fig:Etilde}),
uses a unitary operation on the three-dimensional Hilbert space $\mathscr{H}_3=\mathscr{H}_2 \oplus \mathscr{H}_1$ 
followed by a projective measurement on to $\mathscr{H}_2$ and postselection, as we now explain.
For any $\eta \le \pi(I)$, we first construct the metric operator
\begin{equation}
\label{eq:tildeeta}
\tilde{\eta}:=\frac{1}{\|\eta\|}\eta
\implies\mathcal{G}_\eta = \|\eta\|\mathcal{G}_{\tilde{\eta}},
\end{equation}
and the unitary operator $U_{\tilde{\eta}}\in \mathcal{B}(\mathscr{H}_3)$ that satisfies
\begin{equation}
\label{eq:simulatingFeta}
    \mathcal{G}_{\tilde{\eta}}(\rho)\oplus \bm{0} = PU_{\tilde{\eta}}\sigma U_{\tilde{\eta}}^\dagger P,\, \sigma:=\rho\oplus \bm{0},
    \;\forall \rho\in \mathcal{B}(\mathscr{H}_2),
\end{equation}
where $P$ is the orthogonal projector on~$\mathscr{H}_2$. 
The matrix representation of~$U_{\tilde{\eta}}$ is (see Appendix~\ref{sec:matrixUeta})
\begin{equation}
\label{eq:Ueta}
    \left[U_{\tilde{\eta}}\right]
    = \begin{pmatrix}
    \left[\tilde{\eta}\right]^{\frac12} &\bm{u} \\ -\text{e}^{\text{i} \theta}\bar{\bm{u}}^{\top} &\text{e}^{\text{i}\theta}r 
    \end{pmatrix},
    \text{spec}\left(\tilde{\eta}^{\nicefrac12}\right)=\{1,r\},
    \theta\in [0,2\pi),
\end{equation}
where~$[\;]$ denotes matrix representation,
$\bm{u}$ is the eigenvector of $[\tilde{\eta}]^{\nicefrac12}$ with eigenvalue $r$
and $\| \bm{u}\| = \sqrt{1-r^2}$.
Furthermore,
$\bar{\bm{u}}^{\top}$ is the Hermitian conjugate of the vector $\bm{u}$.
Both~$\theta$ and the global phase of $\bm{u}$ are free parameters.
The qutrit unitary operator $U_{\tilde{\eta}}$ is part of the overall 
simulation procedure (Eq.~\eqref{eq:simulatingFeta}).

Now we explain how an agent 
can sequentially apply each operator in Eq.~\eqref{eq:simulatingFeta}
to simulate $\mathcal{G}_{\eta}$ (Fig.~\ref{fig:flowchart}).
The agent is provided with a description of $2\times2$ matrix~$[\eta]$,
in the logical basis~$\{\ket0,\ket1\}$ 
and a quantum state $\sigma$ (Eq.~\eqref{eq:simulatingFeta}).
The task is to generate the state
$\left({\mathcal{G}_{\eta}(\rho)\oplus \bm{0}}\right)/{{\rm tr}\left(\mathcal{G}_{\eta}(\rho)\oplus \bm{0}\right)}$
with probability ${\rm tr}\left(\mathcal{G}_{\tilde{\eta}}(\rho)\oplus \bm{0}\right)$.
The agent first computes $\left[\tilde{\eta}\right]$ (Eq.~\eqref{eq:tildeeta}) and
$\left[U_{\tilde{\eta}}\right]$ (Eq.~\eqref{eq:Ueta})
and then applies physical operations corresponding to~$U_{\tilde{\eta}}$ on~$\sigma$ followed by projective measurement~$P$ (Eq.~\eqref{eq:simulatingFeta}).
For non-zero measurement outcome,
which occurs with probability ${\rm tr}\left(\mathcal{G}_{\tilde{\eta}}(\rho)\oplus \bm{0}\right)$,
the post-measurement state obtained is
(Eq.~\eqref{eq:tildeeta})
\begin{equation}
   \frac{\mathcal{G}_{\tilde{\eta}}(\rho)\oplus \bm{0}}{{\rm tr}(\mathcal{G}_{\tilde{\eta}}(\rho)\oplus \bm{0})} = \frac{\mathcal{G}_\eta(\rho)\oplus \bm{0}}{{\rm tr}(\mathcal{G}_\eta(\rho)\oplus \bm{0})}.
\end{equation}
The agent discards the state if the measurement outcome is zero.
This concludes the simulation procedure.
The agent may further estimate the success probability ${\rm tr}\left(\mathcal{G}_{\tilde{\eta}}(\rho)\oplus \bm{0}\right)$,
if required, by repeating the simulation procedure on a large number of copies
of $\sigma$ provided to them and then calculating the ratio of non-zero measurement outcomes to the total
number of copies used~\footnote{
The ratio of non-zero measurement outcomes to the total
number of copies approaches the success probability by the law of large numbers~\cite{DKLM05}.
Our setting assumes that multiple copies of the state $\sigma$
are provided to the agent by an external agent who has the knowledge of
$\sigma$, and not prepared by the agent implementing the inner-product changing channel, say, by cloning.}.

In Appendix~\ref{sec:qubitPTsymmetry}, we provide an explicit 
procedure to simulate the dynamics (Eq.~\eqref{eq:UPT}) 
of the qubit PT-symmetric Hamiltonian 
~\cite{BB98},
by sequentially applying the operators in Eq.~\eqref{eq:PTUnitaryoperational} and by using the qutrit 
simulation procedure to implement~$\mathcal{G}_{\eta},\mathcal{G}_{\kappa\eta^{-1}}$.
In Appendix~\ref{sec:dsimulation}, we design a simulation procedure, similar to our qutrit procedure given above, 
for changing the inner product of a $d$-dimensional 
system using a $2d$-dimensional system for any positive integer~$d$. Furthermore, we
use our procedure to simulate the dynamics of a $d$-dimensional PT-symmetric Hamiltonian
by using only $2d$ dimensions, instead of $d^3$ dimensions as required in the Stinespring dilation approach \cite{Sti95}.

\begin{figure}
\begin{center}
\includegraphics[width = \columnwidth]{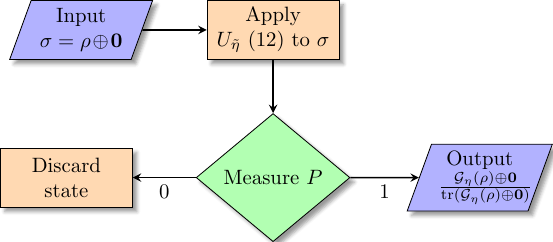}
\end{center}
\vspace{-0.5cm}
\caption{Simulation of the application of $\mathcal{G}_{{\eta}}$ on $\rho$ with success probability
${\rm tr}\left(\mathcal{G}_{\tilde{\eta}}(\rho)\oplus \bm{0}\right)$}. 
\label{fig:flowchart}
\end{figure}

We also design a scheme to verify tomographically whether a prover
can perform an arbitrary change of inner product using our qutrit simulation procedure.
Input to the verification scheme is a threshold function $D_{\rm th}:\mathcal{B}(\mathscr{H}_2)\to (0,1)$ given as a black-box.
The output is `accept' if $\|\mathcal{G}_\eta\oplus \bm{0} - \hat{\mathcal{G}}_\eta\|_{1\to 1}\leq D_{\rm th}(\eta)$ 
or `reject' otherwise, where $\hat{\mathcal{G}}_\eta:\mathcal{B}(\mathscr{H}_3)\to \mathcal{B}(\mathscr{H}_3)$ 
represents a tomographic reconstruction of the qutrit process implemented by the prover, 
$\mathcal{G}_\eta\oplus \bm{0}$ extends the action of $\mathcal{G}_\eta$ to $B(\mathscr{H}_3)$
and $\|\bullet\|_{1 \to 1}$ is the induced Schatten $(1\to 1)$-norm~\cite{Pau03}.
The verifier supplies to the prover a randomly chosen valid $\eta$, a positive integer 
$N$ sufficiently large for the process tomography~\cite{STM11}
and copies of the quantum states $\sigma_i$ encoding $\rho_i$ on demand,
where $\{\rho_i\}$ is chosen based on the tomography procedure in use. 
The prover returns $N$ copies of the qutrit states on which the change of inner product is successful
 as well as the success ratios for each $\rho_i$,
both of which are used by the verifier to reconstruct $\hat{\mathcal{G}_\eta}$.
To ensure that the verifier does not accept the process 
performed by a dishonest prover implementing only
qubit-unitary channels and randomly discarding the system,
it suffices to set the threshold to 
$D_{\rm th}(\eta) = \nicefrac{1}{3}(\lambda_1-\lambda_2)$,
where $\lambda_1>\lambda_2>0$ are the eigenvalues of $\eta$ (see Appendix~\ref{sec:threshold}).

In conclusion, we have three major results.
First, we have operationalized Hilbert-space inner-product change in a way that is both observable and fully compatible with axiomatic quantum mechanics. 
Physically we can understand this inner-product change as 
a lossy quantum operation effecting a change in norm.
This lossy operation
is reminiscent of how superluminality is reconciled by electromagnetic absorption~\cite{SKC93},
with loss in our case forbidding past counterfactual claims.
Consistency of our work is proven using C$^*$ algebra and representations.
Alternatively,
our claims can be verified experimentally by conducting two physically distinct experiments.
One experiment is for the lower-dimensional lossy quantum operation 
and the other experiment is for the higher-dimensional unitary channel with both realizations yielding the same success ratio and 
measurement statistics for a given task.
Our theory fully explains unbroken PT-symmetric quantum mechanics in all its forms as being about changing Hilbert-space inner product and observing its consequences.
Our scheme for simulating qubit PT-symmetric Hamiltonians 
only requires one extra Hilbert-space dimension and no interaction with the environment,
which eliminates the requirements for 
multiple subsystems and entangling operations used in existing schemes~\cite{GS08a,WLG+19,TWY+16,WLG+19,XWZ+19,GZL+21}.
We also show how to simulate $d$-dimensional ($d\geq 2$) PT-symmetric Hamiltonians using
$2d$ dimensions, as opposed to using $d^3$ dimensions in the Stinespring dilation approach.
Our results open possibilities for simulating PT-symmetric dynamics on new experimental platforms, 
such as transmons, where high fidelity qutrit-unitary operations have already been demonstrated~\cite{BRS+20,KYSL20}.

\acknowledgements 
This project is supported by the Government of Alberta and 
by the Natural Sciences and Engineering Research Council of Canada (NSERC).
S.\ K.\ is grateful for a University of Calgary Eyes High International Doctoral Scholarship and an Alberta Innovates Graduate Student Scholarship.
A.\ A.\ acknowledges support through a Killam 2020 Postdoctoral Fellowship.

\begin{appendix}

\section{Constructing a representation of the C\texorpdfstring{$^{*}$}{*} algebra on the Hilbert space with a different inner product}\label{sec:newHilbertspace}
In this section, we show the construction of the Hilbert space $\mathscr{H}_\eta$ with inner product related 
to that of $\mathscr{H}$ by the metric
operator $\eta$,
the construction of a ${}^*$-representation of the C$^{*}$ algebra $\mathcal{A}$ on this new Hilbert space, 
and
finally the representation of states in $\mathcal{S}$ using density operators on $\mathscr{H}_\eta$.

\subsection{Constructing a new Hilbert space from the metric operator}
For a possibly infinte dimensional Hilbert space~$\mathscr{H}$, we denote by $\mathcal{L}(\mathscr{H})$ and $\mathcal{B}(\mathscr{H})$ the algebra of 
linear and bounded linear operators on~$\mathscr{H}$ respectively. We also denote by $\mathcal{D}\left(\mathscr{H}\right) := 
\{\rho \in \mathcal{B}(\mathscr{H}) : \rho \ge 0, \rho^\dagger = \rho, \text{tr}\rho\le 1 \}$ the set 
of density operators acting on~$\mathscr{H}_\eta$. 
\begin{definition}[\cite{Con07}]
The adjoint of an operator $A \in \mathcal{B}(\mathscr{H})$ is the unique operator $A^\dagger\in \mathcal{B}(\mathscr{H})$ satisfying
\begin{equation}
    \braket{\phi|A|\psi} = \overline{\braket{\psi|A^\dagger|\phi}} \;\; \forall \ket{\phi},\ket{\psi} \in \mathscr{H}.
\end{equation} 
The operator $A$ is self-adjoint if $A = A^\dagger$.
\end{definition}
\begin{definition}[\cite{Con07}]
An operator $A \in \mathcal{L}(\mathscr{H})$ is positive-definite if 
\begin{equation}
    \braket{\phi|A|\phi} >0  \;\; \forall \ket{\phi} \in \mathscr{H},\;\ket{\phi} \ne 0.
\end{equation} 
\end{definition}
The following theorem is adapted from the Appendix A of Ref.~\cite{SGH92}.
\begin{theorem}
\label{thm:Heta}
For any Hilbert space $\mathscr{H} = \left(\mathscr{V},\langle\,\vert\,\rangle\right)$ 
and a self-adjoint, positive-definite operator $\eta \in \mathcal{B}(\mathscr{H})$, 
\begin{enumerate}
    \item the sesquilinear form 
    \begin{equation}
    \label{eq:newip}
        \langle\bullet|\bullet \rangle_\eta := \langle\bullet|\eta|\bullet\rangle
    \end{equation}
    is non-degenerate, therefore an inner product on $\mathscr{V}$.
    \item The vector space $\mathscr{V}$ is complete with respect to the norm induced by the inner product 
    $\braket{\bullet|\bullet}'$, therefore $\mathscr{H}_\eta = \left(\mathscr{V},\langle\,\vert\,\rangle_\eta\right)$  is a Hilbert space.
\end{enumerate}
\end{theorem}

\subsection{Constructing a \texorpdfstring{${}^{*}$}{*}-representation on the new Hilbert space}

We now construct a ${}^*$-representation of the algebra $\mathcal{A}$ on the new Hilbert space $\mathscr{H}_\eta$ constructed in Theorem~\ref{thm:Heta}.
In the following, $\mathscr{H}$ and $\mathscr{H}_\eta$ are two Hilbert spaces with their inner product related by
the metric operator $\eta$ as in Theorem~\ref{thm:Heta}, $\mathcal{A}$ is a C${}^*$ algebra of operators and
$\pi:\mathcal{A} \to \mathcal{B}(\mathscr{H})$ is a ${}^*$-representation of $\mathcal{A}$.
We first establish some results required for constructing such a new representation. The following lemma,
which establishes the inverse of the metric operator $\eta$, is adapted from the Appendix A of Ref.~\cite{SGH92}.
\begin{lemma}
Any self-adjoint and positive-definite operator $\eta \in \mathcal{B}(\mathscr{H})$
is invertible. 
Furthermore, the inverse $\eta^{-1}\in \mathcal{B}(\mathscr{H})$ is self-adjoint and positive-definite.
\end{lemma}
We next show that the bounded operator spaces on $\mathscr{H}$ and $\mathscr{H}_\eta$ coincide.
\begin{lemma}\label{lem:boundedopspaces}
    $M\in\mathcal{B}\left(\mathscr{H}\right)$ if and only if $M\in\mathcal{B}\left(\mathscr{H}_\eta\right)$.
\end{lemma}
\begin{proof}
    Let $\|\bullet\|,\|\bullet\|_\eta$ respectively denote the operator norms in~$\mathscr{H}$, $\mathscr{H}_\eta$.
    From Eq.~\eqref{eq:newip}, 
    \begin{equation}\label{eq:twonorms}
        \left\|M\right\|_\eta = \left\|\eta^{\nicefrac{1}{2}}M\eta^{\nicefrac{-1}{2}}\right\|,\, \forall\, M\in\mathcal{B}(\mathscr{H}).
    \end{equation}
    If $M\in \mathcal{B}\left(\mathscr{H}\right)$, then
    \begin{equation}
        \left\|M\right\|_\eta \leq \left\|\eta^{\nicefrac{1}{2}}\right\|\cdot\left\|M\right\|\cdot\left\|\eta^{\nicefrac{-1}{2}}\right\|<\infty
    \end{equation}
    and therefore, $M\in \mathcal{B}\left(\mathscr{H}_\eta\right)$. 
    To verify the reverse implication, note that
    \begin{eqnarray}
        \left\|M\right\| &=& \left\|\eta^{\nicefrac{-1}{2}}\left(\eta^{\nicefrac{1}{2}}M\eta^{\nicefrac{-1}{2}}\right)\eta^{\nicefrac{1}{2}}\right\|\nonumber \\
        &\leq & \left\|\eta^{\nicefrac{-1}{2}}\right\|\cdot \left\|M\right\|_\eta\cdot  \left\|\eta^{\nicefrac{1}{2}}\right\| 
        <\infty,
    \end{eqnarray}
     for all $M\in\mathcal{B}\left(\mathscr{H}_\eta\right) $.
\end{proof}
The next lemma relates ${}^\dagger$ to ${}^\ddagger$, with the latter denoting the adjoint with respect to the inner product $\braket{\bullet|\bullet}'$ of $\mathscr{H}_\eta$.
\begin{lemma}
For any $M \in \mathcal{B}(\mathscr{H})$, $M^\ddagger = \eta^{-1}M^\dagger \eta$. Additionally, $\eta^\ddagger = \eta$.
\end{lemma}
\begin{proof}
By definition of $\ddagger$, we have 
\begin{eqnarray}
\label{eq:ddagger}
  \langle \phi|M|\psi\rangle_\eta = \overline{\langle \psi|M^\ddagger|\phi\rangle_\eta},\,
  &\forall&\, \ket{\psi},\ket{\phi}\in \mathscr{V}, \nonumber\\
  &\forall& \, M\in \mathcal{B}(\mathscr{H}).  
\end{eqnarray}
Using Eq.~\eqref{eq:newip}, 
\begin{eqnarray}
\label{eq:ddaggerformula}
    \langle \phi|M|\psi\rangle_\eta &=  \langle \phi|\eta M|\psi\rangle = \overline{\langle\psi |M^\dagger\eta|\phi\rangle} =  
    \overline{\langle\psi |\eta^{-1}M^\dagger\eta|\phi\rangle_\eta}, \nonumber\\
    &\forall\, \ket{\psi},\ket{\phi}\in \mathscr{V}, \, \forall \, M\in \mathcal{B}\left(\mathscr{H}\right).
\end{eqnarray}
Then $M^\ddagger = \eta^{-1}M^\dagger\eta$ follows from the comparison of Eq.~\eqref{eq:ddagger} and Eq.~\eqref{eq:ddaggerformula},
and $\eta^\ddagger = \eta$ can be obtained by substituting $M=\eta$ in this relation.
\end{proof}

We are now ready for the construction of a ${}^*$-representation of $\mathcal{A}$ on $\mathscr{H}_\eta$.
\begin{theorem}
\label{lem:pieta}
Let $\pi:\mathcal{A} \to \mathcal{B}(\mathscr{H})$ be a ${}^*$-representation of a C${}^*$ algebra $\mathcal{A}$.
Then $\pi_{\eta}: \mathcal{A} \to \mathcal{L}(\mathscr{H}_\eta): A\mapsto \eta^{\nicefrac{-1}{2}}\pi(A) \eta^{\nicefrac{1}{2}} $
is a ${}^*$-representation of $\mathcal{A}$ on $\mathscr{H}_\eta$.
\end{theorem}
\begin{proof}
The range of the map $\pi_{\eta}$ is $\mathcal{B}(\mathscr{H}_\eta)$, which follows from the fact that 
 $\eta^{1/2},\eta^{-1/2}, \pi(A) \in \mathcal{B}\left(\mathscr{H}_\eta\right)$, following Lemma~\ref{lem:boundedopspaces}
 and definition of $\pi$, for all $A\in\mathcal{A}$. 
The map $\pi_{\eta}$ is linear by construction, and it is product preserving
because $\pi_{\eta}(AB) = \eta^{\nicefrac{-1}{2}}\pi(AB)\eta^{\nicefrac{1}{2}} = \eta^{\nicefrac{-1}{2}}\pi(A)\left(\eta^{\nicefrac{1}{2}}\eta^{\nicefrac{-1}{2}}\right)\pi(B)\eta^{\nicefrac{1}{2}} = \pi_{\eta}(A)\pi_{\eta}(B)$. Therefore $\pi_\eta$ is a representation. 
The representation $\pi_{\eta}$ is also ${}^{*}$-preserving, as 
\begin{eqnarray}
    \pi_{\eta}(A^*) &=& \eta^{\nicefrac{-1}{2}}\pi\left(A^*\right)\eta^{\nicefrac{1}{2}} = \eta^{\nicefrac{-1}{2}}\pi(A)^\dagger \eta^{\nicefrac{1}{2}} 
    = \eta^{\nicefrac{1}{2}}\pi(A)^\ddagger \eta^{\nicefrac{-1}{2}} \nonumber \\
    &=&\left(\eta^{\nicefrac{-1}{2}}\pi(A) \eta^{\nicefrac{1}{2}}\right)^\ddagger = \pi_{\eta}(A)^\ddagger.  
\end{eqnarray}
Therefore $\pi_{\eta}$ is a  ${}^{*}$-representation of $\mathcal{A}$ on~$\mathscr{H}_\eta$.
\end{proof}

\subsection{Representing states on the new Hilbert space}
\label{sec:liftmap}
We now characterize the set of states represented by the set of density operators~$\mathcal{D}\left(\mathscr{H}_\eta\right)$
and construct vector representation of the pure states under $\pi_\eta$.
Recall that ${}^{\#}\pi:\rho \mapsto \omega$ such that $\omega(A) = \text{tr}(\rho \pi(A))\  \forall A \in \mathcal{A}$. 
The map ${}^{\#}\pi_\eta$ is defined analogously for the $\pi_\eta$ representation. 
In the following lemma
we show how the density operators acting on $\mathscr{H}$ and $\mathscr{H}_\eta$ are related,
which we further use to prove that the set of states represented under 
$\pi_\eta$ coincides with that represented under $\pi$, namely $\mathcal{S}$.

\begin{lemma}
   An operator $\rho_\eta \in\mathcal{D}\left(\mathscr{H}_\eta\right) $ if and only if
   \begin{equation}\label{eq:rho_eta}
       \rho_\eta = \eta^{\nicefrac{-1}{2}}\rho \eta^{\nicefrac{1}{2}}
   \end{equation}
    for some  $\rho \in\mathcal{D}\left(\mathscr{H}\right)$.
    \end{lemma}
\begin{proof}
Let $\{\ket{e_i}\}$ be an orthonormal basis of $\mathscr{H}$. We first show that $\{\ket{f_i}:= \eta^{\nicefrac{-1}{2}}\ket{e_i}\}$
is an orthonormal basis of $\mathscr{H}_\eta$. The orthonormality of $\{\ket{f_i}\}$ follows from 
$\braket{f_i|f_j}' = \braket{e_i|\eta^{\nicefrac{-1}{2}}\eta\eta^{\nicefrac{-1}{2}} |e_j}' = \delta_{ij}$. 
We prove that $\{\ket{f_i}\}$ is a basis by showing that any $\ket{\phi} \in \mathscr{V}$ can be expressed
as $\ket{\phi} = \sum_i \braket{f_i|\phi}'\ket{f_i}$. Let $\ket{\psi} = \eta^{\nicefrac{1}{2}}\ket{\phi}$.
Then $\ket{\psi} = \sum_i \braket{e_i|\psi}\ket{e_i}$. Now premultiplying by $\eta^{\nicefrac{-1}{2}}$ and substituting
$\ket{f_i} = \eta^{\nicefrac{-1}{2}}\ket{e_i}$, $\ket{\psi} = \eta^{\nicefrac{1}{2}}\ket{\phi}$ yields 
the desired expression $\ket{\phi} = \sum_i \braket{f_i|\phi}'\ket{f_i}$. Note that this sum is convergent because
the set $\{\ket{f_i}\}$ is orthonormal~\cite{Con07}.

To prove the forward implication, we note that Hilbert Schmidt norm of any $\rho \in\mathcal{D}\left(\mathscr{H}\right)$
is finite, i.e.\ $\sum_{i\in B}\|\rho\ket{e_i}\|^2<\infty$, and
${\rm tr}(\rho) = \sum_{i\in B} \langle e_i|\rho |e_i\rangle<\infty$; 
both these properties follow from $\rho$ being a trace-class operator~\cite{Con07}.
For $\rho_\eta\in\mathcal{B}(\mathscr{H}_\eta)$ satisfying Eq.~\eqref{eq:rho_eta}, 
\begin{eqnarray}
    \sum_{i\in B_\eta}\left\|\rho_\eta\ket{f_i}\right\|_\eta^2 &=& 
    \sum_{i\in B}\left\|\eta^{\nicefrac{-1}{2}}\rho\ket{e_i}\right\|_\eta^2 \nonumber\\ 
    &=& \sum_{i\in B}\langle e_i| (\eta^{\nicefrac{-1}{2}}\rho)^\dagger \eta (\eta^{\nicefrac{-1}{2}}\rho)|e_i\rangle \nonumber\\
    &=& \sum_{i\in B}\left\|\rho\ket{e_i}\right\|^2 <\infty.
\end{eqnarray}
Therefore, $\rho_\eta$ is a Hilbert-Schmidt operator in $\mathcal{B}(\mathscr{H}_\eta)$. 
Furthermore, 
\begin{eqnarray}\label{eq:rho_etatrace}
    {\rm tr}(\rho_\eta) &=& \sum_{i\in B_\eta}\langle f_i|\rho_\eta|f_i\rangle_\eta 
 = \sum_{i\in B}\langle e_i|\eta^{\nicefrac{-1}{2}}\eta \rho_\eta\eta^{\nicefrac{-1}{2}}|e_i\rangle \nonumber\\
   & =& \sum_{i\in B}\langle e_i| \rho|e_i\rangle  = {\rm tr}(\rho).
\end{eqnarray}
Therefore,  $\rho_\eta$ is a trace-class operator with $ {\rm tr}(\rho_\eta)\leq 1$ and hence $\rho_\eta \in\mathcal{D}\left(\mathscr{H}_\eta\right) $.
Similarly, the reverse implication that for any $\rho_\eta \in\mathcal{D}\left(\mathscr{H}_\eta\right) $, the operator
$\eta^{\nicefrac{1}{2}}\rho_\eta \eta^{\nicefrac{-1}{2}} \in \mathcal{D}\left(\mathscr{H}\right)$ is proved by 
starting with an orthonormal basis $\{\ket{f_i}\}$ for $\mathscr{H}_\eta$ and observing that 
$\{\eta^{\nicefrac{1}{2}}\ket{f_i}\}$ is an orthornormal basis for $\mathscr{H}$.

\end{proof}

We now show that both $\mathcal{D}(\mathscr{H})$ and  $\mathcal{D}(\mathscr{H}_\eta)$ represent the same set of states, $\mathcal{S}$,
under the respective~$^*$-representations.
\begin{lemma}
    $\rho_\eta\stackrel{\!\!{}^{\#}\!\pi_\eta}{\to} \omega$ if and only if $\rho\stackrel{\!\!{}^{\#}\!\pi}{\to} \omega$, where $\rho,\rho_\eta$
are related by Eq.~\eqref{eq:rho_eta}.

\end{lemma}
\begin{proof}
To prove the forward implication, note $\omega(A) = \text{tr}(\rho \pi(A)) \ \forall A$ by the definition of ${}^\#\pi$. Then
$\text{tr}(\rho \pi(A)) = \text{tr}(\rho_\eta \eta^{\nicefrac{-1}{2}}\pi(A) \eta^{\nicefrac{1}{2}})
=\text{tr}(\rho_\eta\pi_\eta(A))$ using the cyclic property of trace and the definition of $\pi_\eta$ respectively.
Therefore $\omega(A) = \text{tr}(\rho_\eta\pi_\eta(A)) \ \forall A$, and therefore, $\rho_\eta\stackrel{\!\!{}^{\#}\!\pi_\eta}{\to} \omega$.
The reverse implication can be proved by following the same steps in reverse order.

\end{proof}

Finally, we represent pure states in $\mathcal{S}$ by vectors in the Hilbert space $\mathscr{H}_\eta$.
Recall that a state $\omega \in \mathcal{S}$ has a vector-representation $\ket{\psi} \in \mathscr{H}$ under $\pi$ if 
\begin{equation}
    \omega(A) = \braket{\psi | \pi(A)|\psi} \;\; 
    \forall A \in \mathcal{A}.
\end{equation}
We now extend this definition to the representation $\pi_\eta$.
\begin{definition}
\label{def:vecrep}
A state $\omega \in \mathcal{S}$ has a vector representation $\ket{\psi} \in \mathscr{H}_\eta$ under $\pi_\eta$ if 
\begin{equation}
    \omega(A) = \braket{\psi | \pi_\eta(A)|\psi}_\eta \;\; \forall A \in \mathcal{A}.
\end{equation}
\end{definition}

Let the ball $\overline{B_1}(\mathscr{H}) := \{\ket{\psi} \in \mathscr{H}: \sqrt{\braket{\psi|\psi}} \le 1\}$ 
denote the set of normalized and subnormalized vectors in $\mathscr{H}$. 
Recall that the map $\operatorname{lift}:\overline{B_1}(\mathscr{H}) \to \mathcal{D}(\mathscr{H}):
\ket{\psi} \mapsto \ket{\psi}\bra{\psi}$
connects the vector representation $\ket{\psi}$ of a pure state $\omega$ to its density operator representation
$\ket{\psi}\bra{\psi}$ under $\pi$. 
We now construct an analogous map $\operatorname{lift}_\eta: \overline{B_1}(\mathscr{H}_\eta) \to \mathcal{D}\left(\mathscr{H}_\eta\right)$
for the representation $\pi_\eta$, where the ball 
$\overline{B_1}(\mathscr{H}_\eta) := \{\ket{\psi} \in \mathscr{H}_\eta: \sqrt{\braket{\psi|\psi}}_\eta \le 1\}$. 
\begin{definition}
The map $\operatorname{lift}_\eta: \overline{B_1}(\mathscr{H}_\eta) \to \mathcal{D}(\mathscr{H}_\eta)$ is defined to be the map that 
satisfies the following condition: for any state $\omega \in \mathcal{S}$ with vector representation
$\ket{\psi} \in \mathscr{H}_\eta$ and density operator representation $\rho_\eta \in \mathcal{D}(\mathscr{H}_\eta)$,
$\operatorname{lift}_\eta:\ket{\psi} \mapsto \rho_\eta$.
\end{definition}
We now derive the explicit action of $\operatorname{lift}_\eta$.
\begin{lemma}
\label{lem:lifteta}
The map $\operatorname{lift}_\eta$ has action 
$\operatorname{lift}_\eta:\overline{B_1}(\mathscr{H}_\eta)\to \mathcal{D}(\mathscr{H}_\eta):\ket{\psi} \mapsto \ket{\psi}\bra{\psi}\eta$.
\end{lemma}
\begin{proof}
Let $\operatorname{lift}_\eta:\ket{\psi} \mapsto \rho_\eta$ and ${}^\#\pi_\eta:\rho_\eta \mapsto \omega$. 
Following the definition of ${}^\#\pi_\eta$ and  $\operatorname{lift}_\eta$,
\begin{equation}
    \braket{\psi|\pi_\eta(A)|\psi}_\eta = \text{tr}\left(\rho_\eta\pi_\eta(A)\right) = \omega(A) \;\; \forall A.
\end{equation}
As $\braket{\psi|\pi_\eta(I)|\psi}_\eta  = \omega(I)\le 1$, $\ket{\psi} \in \overline{B_1}(\mathscr{H}_\eta)$.
Using Eq.~\eqref{eq:newip}, we get $\braket{\psi|\pi_\eta(A)|\psi}_\eta = \braket{\psi|\eta\pi_\eta(A)|\psi}$. Then using
the cyclic property of the trace, this expectation value can be expressed as
\begin{equation}
    \braket{\psi|\eta\pi_\eta(A)|\psi} = \text{tr}\left(\ket{\psi}\bra{\psi}\eta\pi_\eta(A)\right) \;\; \forall A,
\end{equation}
therefore, $\rho_\eta = \ket{\psi}\bra{\psi}\eta \in \mathcal{D}(\mathscr{H}_\eta)$.
This leads to the desired action of $\operatorname{lift}_\eta$.

\end{proof}

\section{{{Quantum operation} for implementing the changing inner product}}\label{sec:channelEeta}
In this section, we construct the quantum operation that implements the change in inner product by $\eta \le \pi(I)$. 
Change in inner product is defined by the identity isomorphism $\mathcal{I}_\eta:\mathscr{H}\to\mathscr{H}_\eta$ 
(see Fig.~1c in main text).
We now show how  $\mathcal{I}_\eta$ is extended to $\mathcal{B}(\mathscr{H}_\eta)$ through the map $\mathcal{E}_\eta$ defined below.
\begin{theorem}
Let $\eta \le \pi(I)$ and $\mathcal{E}_\eta:\mathcal{B}(\mathscr{H})\to \mathcal{L}(\mathscr{H}_\eta):M \mapsto M\eta$.
Then
\begin{enumerate}
    \item $\operatorname{range}(\mathcal{E}_\eta) \subseteq \mathcal{B}(\mathscr{H}_\eta)$,
    \item $\mathcal{E}_\eta$ satisfies the following commutative diagram:
    \begin{equation}
    \xymatrix{
    \mathcal{D}(\mathscr{H}) \ar[r]^{\mathcal{E}_\eta} &  \mathcal{D}(\mathscr{H}_\eta)   \\
    \overline{B_1}(\mathscr{H}) \ar[u]^{\operatorname{lift}} \ar[r]^{\mathcal{I}_\eta} 
      & \overline{B_1}(\mathscr{H}_\eta) \ar[u]^{\operatorname{lift}_\eta} }
    \end{equation}
    \item $\mathcal{E}_\eta$ is a quantum operation.
\end{enumerate}
\end{theorem}
\begin{proof}
To prove Statement 1, note that for any $M \in \mathcal{B}\left(\mathscr{H}\right)$, the operator $M\eta \in  \mathcal{B}\left(\mathscr{H}_\eta\right)$ because
\begin{equation}
    \|M\eta\|_\eta = \|\eta^{\nicefrac{1}{2}}\left(M\eta\right)\eta^{\nicefrac{-1}{2}}\|\leq \|\eta^{\nicefrac{1}{2}}\|^2\cdot\|M\|\le \infty,
\end{equation}
where the first equality follows from Eq.~\eqref{eq:twonorms}. 
The commutative diagram in Statement 2 follows immediately from the action of lift${}_\eta$ map
in Lemma \ref{lem:lifteta}. 

We now show that $\mathcal{E}_\eta$ is a valid quantum operation,
i.e.\ a completely-positive, trace non-increasing map.
To prove the positivity of $\mathcal{E}_\eta$,
let $M \ge 0$, so that it can be expressed as $M = A A^\dagger$~\cite{Con07}. 
Then 
$\mathcal{E}_\eta(M) = A A^\dagger \eta$, which can expressed as
$\mathcal{E}_\eta(M) = BB^\ddagger$ with $B = A\eta^{\frac12}$. 
Therefore $\mathcal{E}_\eta(M) \in \mathcal{B}(\mathscr{H}_\eta)$  
is positive if $M$ is positive, which proves the positivity of $\mathcal{E}_\eta$.

Complete positivity of $\mathcal{E}_\eta$ can be proven by showing that the map 
$\mathcal{E}_\eta \otimes \mathscr{I}_k:\mathcal{B}\left(\mathscr{H}\right)\otimes\mathcal{B}(\mathbb{C}^k) \to\mathcal{B}\left(\mathscr{H}_\eta\right)\otimes\mathcal{B}(\mathbb{C}^k)  $ 
is positive,
for every positive integer $k$,
where $\mathscr{I}_k$ denotes the identity map on $\mathcal{B}(\mathbb{C}^k)$. 
The action of the new map is given by
$\left[\mathcal{E}_\eta \otimes \mathscr{I}_k\right]\left(N\right)= N \left(\eta \otimes I_k\right)$,
with $I_k \in \mathcal{B}(\mathbb{C}^k)$
the identity operator. 
Operator 
$\left[\mathcal{E}_\eta \otimes \mathscr{I}_k\right](N)\in\mathcal{B}\left(\mathscr{H}_\eta\right)\otimes\mathcal{B}(\mathbb{C}^k) $
because $\|\eta \otimes I_k\| = \|\eta\|_\eta\cdot\|I_k\| = \|\eta\|_\eta$. 
For proving positivity,
let $N\in \mathcal{B}(\mathscr{H})\otimes \mathcal{B}(\mathbb{C}^k)$ be a positive operator,
so that $N = C C^\dagger$. Then $\left[\mathcal{E}_\eta \otimes \mathscr{I}_k\right](N) = C C^\dagger (\eta \otimes I_k)$, 
which can be expressed as $\left[\mathcal{E}_\eta \otimes \mathscr{I}_k\right] (N) =D D^\ddagger$ with
$D = C(\eta^{\frac12}\otimes I_k)$,
thereby proving positivity of $\left[\mathcal{E}_\eta \otimes \mathscr{I}_k\right](N)$ and consequently positivity of 
$\mathcal{E}_\eta \otimes \mathscr{I}_k$.

To prove that $\mathcal{E}_\eta$ is trace non-increasing, let $\rho \in \mathcal{D}(\mathscr{H})$
and note that trace is independent of the inner product (Eq.~\eqref{eq:rho_etatrace}).
We have $\text{tr}(\rho\eta) = \text{tr}(\eta^{\nicefrac{1}{2}}\rho\eta^{\nicefrac{1}{2}})$
and 
\begin{equation}\label{eq:tracedecreasing}
    \text{tr}(\eta^{\nicefrac{1}{2}}\rho\eta^{\nicefrac{1}{2}}) = 
    \text{tr}|\eta^{\nicefrac{1}{2}}\rho\eta^{\nicefrac{1}{2}}|
    \le \|\eta^{\nicefrac{1}{2}}\|^2\text{tr}|\rho| \le \text{tr}(\rho),
\end{equation}
where $|M| = \sqrt{\left(M^\dagger M\right)}$ and we used $|M| = M$ for any $M \ge 0$. 
The first inequality in Eq.~\eqref{eq:tracedecreasing} is a property of the trace norm~\cite{Con07}, and the last inequality in Eq.~\eqref{eq:tracedecreasing} follows from the fact that $\eta\leq \pi(I)$ and therefore, $\left\|\eta^{\nicefrac{1}{2}}\right\|^2\leq 1$.

\end{proof}

\section{Transformation of the Operators under Changing the Inner Product}\label{sec:commutation}
An inner product changing channel could modify the 
commutation relations between the operators. 
In this section, we demonstrate such a change
with an explicit example of a qubit system undergoing an inner product change.
Consider a qubit system undergoing change in inner product by
\begin{equation}
    \eta = \frac{1}{1 + r\sin\phi} \begin{pmatrix}
1 & -\text{i}r\sin\phi \\
\text{i}r\sin\phi & 1
 \end{pmatrix}, \quad 0\le r<1.
\end{equation}
The Pauli operators $X,Y,Z \in \mathcal{B}(\mathscr{H})$ acting on the original
Hilbert space along with the identity operator $I_2\in \mathcal{B}(\mathscr{H})$ generate the $\mathfrak{u}(2)$
algebra.
These operators transform according to Eq.~(5) in the main text
following the inner product change by $\eta$.
This transformation is given by the map $\mathcal{E}^{\rm op}_\eta$.
The transformed operators satisfy the commutation relations
\begin{align}
    &\left[\mathcal{E}^{\rm op}_\eta\left(X\right),\mathcal{E}^{\rm op}_\eta\left(Y\right)\right] =  2\text{i}a\ \mathcal{E}^{\rm op}_\eta\left(Z\right),\nonumber \\ 
    &\left[ \mathcal{E}^{\rm op}_\eta\left(I_2\right), \mathcal{E}^{\rm op}_\eta\left(Z\right) \right] = 2\text{i}(1-a)\  \mathcal{E}^{\rm op}_\eta\left(X\right), \nonumber \\ 
    & \left[\mathcal{E}^{\rm op}_\eta\left(Y\right),\mathcal{E}^{\rm op}_\eta\left(Z\right)\right] =  2\text{i}a\ \mathcal{E}^{\rm op}_\eta\left(X\right),\nonumber \\ 
    &\left[ \mathcal{E}^{\rm op}_\eta\left(I_2\right), \mathcal{E}^{\rm op}_\eta\left(X\right) \right] = -2\text{i}(1-a)\ \mathcal{E}^{\rm op}_\eta\left(Z\right), 
    \nonumber \\
    &\left[\mathcal{E}^{\rm op}_\eta\left(Z\right),\mathcal{E}^{\rm op}_\eta\left(X\right)\right] =  -2\text{i}(1-a)\ \mathcal{E}^{\rm op}_\eta\left(I_2\right)+2\text{i}a\mathcal{E}^{\rm op}_\eta\left(Y\right),\nonumber \\ 
    &\left[\mathcal{E}^{\rm op}_\eta\left(I_2\right),\mathcal{E}^{\rm op}_\eta\left(Y\right)\right] = \bm{0},
\end{align}
where $a = 1/(1+r\sin \phi)$. These commutation relations are different from 
those of $\mathfrak{u}(2)$ algebra for $r \ne 0$, or equivalently $a \ne 1$.

\section{{Matrix representation of the qutrit unitary operator that simulates
 change in inner product of a qubit system}}\label{sec:matrixUeta}
In this section, we derive the matrix representation of the qutrit unitary operator $U_{\tilde{\eta}}$
(see Eq.~(12) in main text) employed in the 
simulation of the change in inner product of a qubit system. Equation~(10) in the main text
requires that $PU_{\tilde{\eta}}P = \tilde{\eta}^{\nicefrac{1}{2}}$, so that
$U_{\tilde{\eta}}$ can be expressed as
\begin{equation}\label{eq:Utildeeta}
    [U_{\tilde{\eta}}] = \begin{pmatrix}
    \left[\tilde{\eta}\right]^{\frac12} & \bm{u} \\ {\bar{\bm{v}}}^\top &r\text{e}^{\text{i}\theta} 
    \end{pmatrix}
\end{equation}
for some vectors $\bm{u},\bm{v} \in \mathbb{C}^2$, a number $r \in [0,1]$ and a phase $\theta \in [0,2\pi)$.
The unitarity conditions $U_{\tilde{\eta}}^\dagger U_{\tilde{\eta}} = U_{\tilde{\eta}} U_{\tilde{\eta}}^\dagger = I_3$ lead to
\begin{align}
    \left[\tilde{\eta}\right] + {\bm{u}}\bar{\bm{u}}^\top  = I_2,
    \label{eq:unitary1}\\
    \left[\tilde{\eta}\right]^{\frac12}\bm{u} + r\text{e}^{{\rm i}\theta}\bm{v}=0,\label{eq:unitary2}\\   
    \bar{\bm{u}}^\top{\bm{u}}+r^2 =  1 \implies \|\bm{u}\| = 1-r^2 \label{eq:unitary3}.
\end{align}
Postmultiplying Eq.~\eqref{eq:unitary1} by $\bm{u}$ and substituting Eq.~\eqref{eq:unitary3} yields
$\left[\tilde{\eta}\right]\bm{u} = r^2\bm{u}$, which implies that $\bm{u}$ is the eigenvector
of $\left[\tilde{\eta}\right]^{\nicefrac{1}{2}}$ with eigenvalue $r$. Then Eq.~\eqref{eq:unitary2}
yields $\bm{v} = -\text{e}^{-{\rm i}\theta}\bm{u}$ as desired, 
with the global phase of $\bm{u}$ and $\theta$ being the free parameters.

\section{{Simulation of a qubit PT-symmetric Hamiltonian using single qutrit}}\label{sec:qubitPTsymmetry}

We now design a qutrit procedure that simulates the dynamics of a qubit Hamiltonian with unbroken PT symmetry.
Our design is based on the qutrit procedure for simulating the change in inner product of a qubit system, 
provided in the main text  (Fig.~4). 
We illustrate our Hamiltonian-simulation procedure using the PT-symmetric Hamiltonian $H_{\rm PT}$ from~\cite{BBJ02}. 

The matrix form of  $H_{\rm PT}$ is
\begin{equation}
\label{eq:qubitHPT}
    \left[H_{\rm PT}\right] = \begin{pmatrix}
    r\text{e}^{{\rm i}\phi} & s \\ s & r\text{e}^{-{\rm i}\phi}\end{pmatrix}, \quad s>r\sin\phi\ge 0.
\end{equation}
Dynamics generated by $H_{\rm PT}$ is denoted by the operator $\mathcal{U}_{\rm PT}$ (Eq.~(6) in main text),
\begin{equation}
    \rho \stackrel{\mathcal{U}_{\rm PT}}{\mapsto} \kappa U_{\rm{PT}}\rho U^\dagger_{\rm{PT}}, \; U_{\rm{PT}} := \text{e}^{-\text{i}H_{\rm {PT}}t/\hbar},\, \kappa = \frac{1}{\|\eta_2^{-1}\|}.
\end{equation}
As proved in the main text, $\mathcal{U}_{\rm PT}$ can be expressed as a sequence of operations acting exclusively on 
$ \mathcal{B}(\mathscr{H}_2)$,
\begin{equation}~\label{eq:simulatingUPT}
     \mathcal{U}_{\rm{PT}}= {\mathcal{G}}_{\kappa \eta^{-1}_2} \circ (\mathcal{R}_{\kappa\eta^{-1}_2} \circ \widetilde{\mathcal{U}}_{\rm{PT}} \circ \mathcal{R}_{\eta_2})\circ \mathcal{G}_{\eta_2},
\end{equation}
where
$\mathcal{R}_{\kappa\eta^{-1}} \circ \widetilde{\mathcal{U}}_{\rm{PT}} \circ \mathcal{R}_{\eta}: \mathcal{B}(\mathscr{H}_2)\to \mathcal{B}(\mathscr{H}_2)$
is the channel with unitary Kraus operator $\eta^{\nicefrac{1}{2}}U_{\rm PT}\eta^{\nicefrac{-1}{2}}$ and
the Kraus operator is generated by self-adjoint Hamiltonian $ h_{\rm PT} = \eta^{\nicefrac{1}{2}}H_{\rm PT}\eta^{\nicefrac{-1}{2}}$.

Our qutrit procedure for simulating $\mathcal{U}_{\rm PT}$ involves implementing each operation in Eq.~\eqref{eq:simulatingUPT}
using qutrit unitaries and measurements, as we now explain through Steps 1-4.
The input to the simulation procedure is $\rho \in \mathcal{B}(\mathscr{H}_2)$ 
embedded as $\sigma:= \rho\oplus \bm{0} \in \mathcal{B}(\mathscr{H}_3)$ and a time $t>0$.
The output of the procedure is the state $\frac{U_{\rm{PT}}\rho U^\dagger_{\rm{PT}}}{{\rm tr}\left( U_{\rm{PT}}\rho U^\dagger_{\rm{PT}}\right)}$
with probability $\frac{1}{\|\eta_2^{-1}\|}{\rm tr}\left( U_{\rm{PT}}\rho U^\dagger_{\rm{PT}}\right).$
The simulation steps are
\begin{enumerate}
    \item Calculate the metric operator and its inverse:
    The agent calculates $\eta_2$, $\eta_2^{-1}$ satisfying the quasi-Hermiticity condition $H^\dagger_{\rm PT} = \eta_2H_{\rm PT}\eta^{-1}_2$.
    A choice of $\eta_2$ and, therefore $\eta_2^{-1}$, is
\begin{equation}
    \left[\eta_2\right] = \frac{1}{s+r\sin\phi}\begin{pmatrix}
    s & -{\rm i}r\sin\phi \\{\rm i}r\sin\phi & s \end{pmatrix},
\end{equation}
\begin{equation}
    \left[\eta_2^{-1}\right] = \frac{1}{s-r\sin\phi}\begin{pmatrix}
    s & {\rm i}r\sin\phi \\-{\rm i}r\sin\phi & s \end{pmatrix},
\end{equation}
with $\|\eta_2\|=1$ and $\|\eta_2^{-1}\| =(s+r\sin\phi)/(s-r\sin\phi)$.

\item Simulate change in inner product by $\eta_2$: 
Agent implements the qutrit procedure (Fig.~4) to simulate $\mathcal{G}_{{\eta}_2}$ 
by setting $\tilde{\eta} = \eta_2$ and for a single copy of $\sigma$.
A choice of the qutrit unitary~$U_{{\eta}_2}$ (see Eq.~\eqref{eq:Utildeeta})
simulating the action of  $\mathcal{G}_{\eta_2}$ is
\begin{equation}\label{eq:Ueta2}
    U_{{\eta}_2} = \begin{pmatrix}
    (1+q)/2 & -{\rm i}(1-q)/2 & p \\
    {\rm i}(1-q)/2 & (1+q)/2 & -{\rm i}p \\
    -p & -{\rm i}p & q
    \end{pmatrix},
\end{equation}
where
\begin{equation}
    \quad q = \sqrt{\frac{s-r\sin\phi}{s+r\sin\phi}} = \frac{1}{\sqrt{\|\eta_2^{-1}\|}}, \  p = \sqrt{\frac{r\sin\phi}{s+r\sin\phi}}.
\end{equation}
The output of this step is the qutrit state 
$\frac{\eta_2^{\nicefrac{1}{2}}\rho\eta_2^{\nicefrac{1}{2}}}{{\rm tr}\left(\eta_2^{\nicefrac{1}{2}}\rho\eta_2^{\nicefrac{1}{2}}\right)}\oplus\bm{0} $
with probability ${\rm tr}\left(\eta_2^{\nicefrac{1}{2}}\rho\eta_2^{\nicefrac{1}{2}}\right)$.

\item Simulate the unitary evolution generated by $h_{\rm PT}$:
Agent calculates $h_{\rm PT}$ embedded in $\mathcal{B}(\mathscr{H}_3)$,
\begin{equation}
    \qquad\; [h_{\rm PT} \oplus \bm{0}] = \begin{pmatrix}    
    r\cos \phi & \sqrt{s^2 - r^2\sin^2\phi} & 0\\
    \sqrt{s^2 - r^2\sin^2\phi} & r\cos \phi & 0\\
    0 & 0 & 0\end{pmatrix},
\end{equation}
and implements the qutrit unitary operator $e^{-{\rm i}(h_{\rm PT}\oplus \bm{0})t}$,
which is equivalent to simulating the channel $
(\mathcal{R}_{\kappa\eta^{-1}_2} \circ \widetilde{\mathcal{U}}_{\rm{PT}} \circ \mathcal{R}_{\eta_2})$
in Eq.~\eqref{eq:simulatingUPT}. 
The output of this deterministic step is the qutrit state $\left(e^{-{\rm i}h_{\rm PT}t} \frac{\left(\eta_2^{\nicefrac{1}{2}}\rho\eta_2^{\nicefrac{1}{2}}\right)}
{{\rm tr}\left(\eta_2^{\nicefrac{1}{2}}\rho\eta_2^{\nicefrac{1}{2}}\right)}e^{{\rm i}h_{\rm PT}t}\right)\oplus\bm{0} $, 
provided Step 2 is successful.
\item Simulate change in inner product by $\kappa\eta_2^{-1}$:
Agent applies the qutrit procedure (Fig.~4)
to simulate $\mathcal{G}_{\kappa\eta_2^{-1}}$,
by setting $\tilde{\eta} = \kappa\eta^{-1}_2$.
Note that we have $\kappa = \frac{1}{\|\eta_2^{-1}\|} = q^2$ (Eqs.~\eqref{eq:simulatingUPT},~\eqref{eq:Ueta2}).
A choice of $U_{\kappa\eta^{-1}_2} $ is
\begin{equation}
    U_{\kappa{\eta}_2^{-1}} = \begin{pmatrix}
    (1+q)/2 & {\rm i}(1-q)/2 & p \\
    -{\rm i}(1-q)/2 & (1+q)/2 & {\rm i}p \\
    -p & {\rm i}p & q
    \end{pmatrix}.
\end{equation}
The output of this procedure is the qutrit state 
$\frac{\left(\eta^{\nicefrac{-1}{2}}_2e^{-{\rm i}h_{\rm PT}t}\eta_2^{\nicefrac{1}{2}}\rho\eta_2^{\nicefrac{1}{2}}e^{{\rm i}h_{\rm PT}t}\eta^{\nicefrac{-1}{2}}_2\right)}
{{\rm tr}\left(\eta^{\nicefrac{-1}{2}}_2e^{-{\rm i}h_{\rm PT}t}\eta_2^{\nicefrac{1}{2}}\rho\eta_2^{\nicefrac{1}{2}}e^{{\rm i}h_{\rm PT}t}\eta^{\nicefrac{-1}{2}}_2\right)}\oplus\bm{0}  = \frac{U_{\rm PT} \rho U_{\rm PT}^\dagger}{{\rm tr}\left( U_{\rm PT} \rho U_{\rm PT}^\dagger\right)} \oplus \bm{0}$ 
with probability  $\frac{{\rm tr}\left( U_{\rm PT} \rho U_{\rm PT}^\dagger\right)}{\|\eta_2^{-1}\|{\rm tr} \left(\eta_2^{\nicefrac{1}{2}}\rho\eta_2^{\nicefrac{1}{2}}\right)}$.
\end{enumerate}
Therefore, the output of the simulation procedure is the state 
$\frac{U_{\rm PT} \rho U_{\rm PT}^\dagger}{{\rm tr}\left( U_{\rm PT} \rho U_{\rm PT}^\dagger\right)} \oplus \bm{0}$
with success probability given by the combined probability of success in Steps 2,4, which is equal to
$\frac{1}{\|\eta_2^{-1}\|}{{\rm tr}\left( U_{\rm PT} \rho U_{\rm PT}^\dagger\right)}$.

\section{Simulation of change in inner product and PT-symmetric dynamics 
of a \texorpdfstring{$d$}{d}-dimensional system}~\label{sec:dsimulation}
We first explain a simulation procedure
to change the inner product of a $d$-dimensional system using $2d$ dimensions.
We assume that the algebra $\mathcal{A}$ of the system is represented on a $d$-dimensional Hilbert space $\mathscr{H}_d^{(s)}$ by $\pi$.
Similar to the qutrit simulation procedure explained in main text,
the agent simulating $\mathcal{G}_\eta$ for $\eta\leq\pi(I)$ first 
constructs the metric operator $\tilde{\eta} =\frac{1}{\|\eta\|}\eta$
and the unitary operator $U_{\tilde{\eta}}\in\mathcal{B}(\mathscr{H}_d^{(s)} \oplus \mathscr{H}_d^{(a)} )$
satisfying 
\begin{equation}\label{eq:Getad}
    \mathcal{G}_{\tilde{\eta}}(\rho)\oplus \bm{0} = PU_{\tilde{\eta}}\sigma U_{\tilde{\eta}}^\dagger P,\, \sigma:=\rho\oplus \bm{0},
    \;\forall \rho\in \mathcal{B}(\mathscr{H}_d^{(s)}),
\end{equation}
where $\bm{0}$ denotes the zero operator in $\mathcal{B}(\mathscr{H}_d^{(a)})$.
The matrix representation of a choice of $U_{\tilde{\eta}}$ is
\begin{equation}
\label{eq:Uetad}
    \left[U_{\tilde{\eta}}\right]
    = \begin{pmatrix}
    \left[\tilde{\eta}\right]^{\frac12} & \left[1-\tilde{\eta}\right]^{\frac12} \\ \left[1-\tilde{\eta}\right]^{\frac12} & - \left[\tilde{\eta}\right]^{\frac12}
    \end{pmatrix}.
\end{equation}
Agent then implements  $U_{\tilde{\eta}}$
followed by projective measurement and postselection on to the subspace $\mathscr{H}_d^{(s)}$.
All steps of the simulation procedure are similar to the qutrit simulation procedure for changing inner product
explained in the main text.

We now discuss how this simulation procedure for changing the inner product 
can be used for simulating PT-symmetric dynamics in $d$ dimensions
using a $2d$-dimensional system. Similar to the $d=2$ case discussed in Sec.\,\ref{sec:qubitPTsymmetry}, 
the input to the simulation procedure is $\rho \in \mathcal{B}(\mathscr{H}_d^{(s)})$ 
embedded as $\sigma:= \rho\oplus \bm{0} \in \mathcal{B}(\mathscr{H}_d^{(s)}\oplus\mathscr{H}_d^{(a)})$ and a time $t>0$.
The simulation steps are as follows:
\begin{enumerate}
    \item The agent calculates $\eta$ satisfying the quasi-Hermiticity condition $H^\dagger_{\rm PT} = \eta H_{\rm PT}\eta^{-1}$ with $\|\eta\|=1$.
    \item The agent implements the procedure described above to simulate $\mathcal{G}_{\eta}$ 
    by setting $\tilde{\eta} = \eta$ and
    for a single copy of $\sigma$ (Eq.~\eqref{eq:Getad}). 
    \item The agent calculates $h_{\rm PT} = \eta^{\nicefrac{1}{2}}H_{\rm PT}\eta^{-\nicefrac{1}{2}}$ 
    embedded in $\mathcal{B}(\mathscr{H}_d^{(s)}\oplus\mathscr{H}_d^{(a)})$,
    and implements the unitary operator $e^{-{\rm i}(h_{\rm PT}\oplus \bm{0})t}$.
    \item The agent applies the procedure described above to simulate $\mathcal{G}_{\kappa\eta^{-1}}$, 
    by setting $\tilde{\eta} = \kappa\eta^{-1}$ such that $\kappa = \frac{1}{\|\eta^{-1}\|}$ (Eq.~\eqref{eq:Getad}).
\end{enumerate}
The output of this procedure is the state $\frac{U_{\rm{PT}}\rho U^\dagger_{\rm{PT}}}{{\rm tr}\left( U_{\rm{PT}}\rho U^\dagger_{\rm{PT}}\right)}$
with probability $\frac{1}{\|\eta_2^{-1}\|}{\rm tr}\left( U_{\rm{PT}}\rho U^\dagger_{\rm{PT}}\right).$

\section{{Additional details on the verification scheme}}\label{sec:threshold}
We now prove a threshold distance $D_{\rm th}$ for the tomographic verification scheme, for the 
qutrit procedure simulating the change in inner product by an arbitrary $\eta$, provided in the main text.
The scheme allows
a verifier to distinguish an honest prover implementing the operation $\mathcal{G}_\eta$ 
from a dishonest prover failing to implement the same.
We assume that the dishonest prover implements only unitary operations, on the qubit subspace, drawn from the set  
$\{U_j \oplus \bm{1}: U_j\in\mathcal{B}\left(\mathscr{H}_2\right) \}$, where each $U_j$ is selected with probability $p_j$, 
and the system is discarded with probability $p:=1-\sum_jp_j<1$.
The quantum operation implemented by the dishonest prover is given by 
\begin{equation}
    \hat{\mathcal{G}}_\eta(\bullet) = \sum_jp_j \left(U_j \oplus \bm{1}\right)^\dagger\bullet \left(U_j \oplus \bm{1}\right).
\end{equation}
We now derive a lower bound for the induced Schatten $(1\to 1)$-norm distance~\cite{Pau03} between the 
inner-product changing operation~$\mathcal{G}_\eta\oplus \bm{0}$  
and the implemented operation $\hat{\mathcal{G}}_\eta$.
Note that 
\begin{equation}
    \|\mathcal{G}_\eta\oplus\bm{0}-\hat{\mathcal{G}}_\eta\|_{1\to 1} = \max_{T\in\mathcal{B}\left(\mathscr{H}_3\right)}\frac{  \|\mathcal{G}_\eta\oplus\bm{1}(T)-\hat{\mathcal{G}}_\eta(T)\|_{\rm tr}}{\|T\|_{\rm tr}},
\end{equation}
 where  $\|T\|_{\rm tr} = {\rm tr}\left(\sqrt{T^\dagger T}\right)$.
 For $T = \frac{1}{3}I_3$, $\|T\|_{\rm tr} = 1$. 
 Therefore,
 \begin{align}
      \|\mathcal{G}_\eta\oplus\bm{0}-\hat{\mathcal{G}}_\eta\|_{1\to 1} 
      &\ge 
      \frac{1}{3}\|\mathcal{G}_\eta\oplus \bm{0}(I_3)-\hat{\mathcal{G}}_\eta(I_3)\|_{\rm tr} \nonumber\\
      &= 
      \frac{1}{3}\|\eta\oplus\bm{1}-\sum_jp_j(U_j\oplus\bm{1})^\dagger I_3(U_j\oplus\bm{1})\|_{\rm tr} \nonumber\\
      &= 
       \frac{1}{3}\|\eta\oplus\bm{1}-(1-p)I_3\|_{\rm tr}.
 \end{align}
 We assume that the eigenvalues of $\eta$ are denoted by $\lambda_1,\lambda_2$.
 The eigenvalues satisfy  $1\geq\lambda_1>\lambda_2>0$ for any non-trivial $\eta$, i.e.\ $\eta\neq I_3$.
 The trace distance $\|\eta\oplus\bm{1}-(1-p)I_3\|_{\rm tr} = \vert \lambda_1-(1-p)\vert + \vert\lambda_2-(1-p)\vert$.
 For any $p\in[0,1)$, it can be further verified that $\|\eta\oplus\bm{1}-(1-p)I_3\|_{\rm tr}\geq {(\lambda_1-\lambda_2)}.$
 Therefore, 
 \begin{equation}
     D_{\rm th} :=  \frac{(\lambda_1-\lambda_2)}{3} \leq \|\mathcal{G}_\eta\oplus\bm{0}-\hat{\mathcal{G}}_\eta\|_{1\to 1}.
 \end{equation}
  The above given value for $D_{\rm th}$ allows the verifier to distinguish an honest prover from a dishonest one, provided the honest prover 
  implements the operation $\mathcal{G}_{\eta}$ with error less than $D_{\rm th}$, where error is quantified by the induced Schatten $(1\to 1)$-norm.

\end{appendix}
\bibliography{ref.bib}

\begin{thebibliography}{73}%
\makeatletter
\providecommand \@ifxundefined [1]{%
 \@ifx{#1\undefined}
}%
\providecommand \@ifnum [1]{%
 \ifnum #1\expandafter \@firstoftwo
 \else \expandafter \@secondoftwo
 \fi
}%
\providecommand \@ifx [1]{%
 \ifx #1\expandafter \@firstoftwo
 \else \expandafter \@secondoftwo
 \fi
}%
\providecommand \natexlab [1]{#1}%
\providecommand \enquote  [1]{``#1''}%
\providecommand \bibnamefont  [1]{#1}%
\providecommand \bibfnamefont [1]{#1}%
\providecommand \citenamefont [1]{#1}%
\providecommand \href@noop [0]{\@secondoftwo}%
\providecommand \href [0]{\begingroup \@sanitize@url \@href}%
\providecommand \@href[1]{\@@startlink{#1}\@@href}%
\providecommand \@@href[1]{\endgroup#1\@@endlink}%
\providecommand \@sanitize@url [0]{\catcode `\\12\catcode `\$12\catcode `\&12\catcode `\#12\catcode `\^12\catcode `\_12\catcode `\%12\relax}%
\providecommand \@@startlink[1]{}%
\providecommand \@@endlink[0]{}%
\providecommand \url  [0]{\begingroup\@sanitize@url \@url }%
\providecommand \@url [1]{\endgroup\@href {#1}{\urlprefix }}%
\providecommand \urlprefix  [0]{URL }%
\providecommand \Eprint [0]{\href }%
\providecommand \doibase [0]{https://doi.org/}%
\providecommand \selectlanguage [0]{\@gobble}%
\providecommand \bibinfo  [0]{\@secondoftwo}%
\providecommand \bibfield  [0]{\@secondoftwo}%
\providecommand \translation [1]{[#1]}%
\providecommand \BibitemOpen [0]{}%
\providecommand \bibitemStop [0]{}%
\providecommand \bibitemNoStop [0]{.\EOS\space}%
\providecommand \EOS [0]{\spacefactor3000\relax}%
\providecommand \BibitemShut  [1]{\csname bibitem#1\endcsname}%
\let\auto@bib@innerbib\@empty
\bibitem [{\citenamefont {Scholtz}\ \emph {et~al.}(1992)\citenamefont {Scholtz}, \citenamefont {Geyer},\ and\ \citenamefont {Hahne}}]{SGH92}%
  \BibitemOpen
  \bibfield  {author} {\bibinfo {author} {\bibfnamefont {F.~G.}\ \bibnamefont {Scholtz}}, \bibinfo {author} {\bibfnamefont {H.~B.}\ \bibnamefont {Geyer}},\ and\ \bibinfo {author} {\bibfnamefont {F.~J.~W.}\ \bibnamefont {Hahne}},\ }\bibfield  {title} {\bibinfo {title} {{Quasi-Hermitian operators in quantum mechanics and the variational principle}},\ }\href@noop {} {\bibfield  {journal} {\bibinfo  {journal} {Ann. Phys.}\ }\textbf {\bibinfo {volume} {213}},\ \bibinfo {pages} {74} (\bibinfo {year} {1992})}\BibitemShut {NoStop}%
\bibitem [{\citenamefont {Narici}\ and\ \citenamefont {Beckenstein}(2010)}]{Narici2010}%
  \BibitemOpen
  \bibfield  {author} {\bibinfo {author} {\bibfnamefont {L.}~\bibnamefont {Narici}}\ and\ \bibinfo {author} {\bibfnamefont {E.}~\bibnamefont {Beckenstein}},\ }\href {https://doi.org/10.1201/9781584888673} {\emph {\bibinfo {title} {Topological Vector Spaces}}}\ (\bibinfo  {publisher} {Chapman and Hall/{CRC}},\ \bibinfo {year} {2010})\BibitemShut {NoStop}%
\bibitem [{\citenamefont {Blanchard}\ and\ \citenamefont {Br{\"u}ning}(2003)}]{BB03}%
  \BibitemOpen
  \bibfield  {author} {\bibinfo {author} {\bibfnamefont {P.}~\bibnamefont {Blanchard}}\ and\ \bibinfo {author} {\bibfnamefont {E.}~\bibnamefont {Br{\"u}ning}},\ }\href {https://doi.org/10.1007/978-3-319-14045-2} {\emph {\bibinfo {title} {{Mathematical Methods in Physics: Distributions, Hilbert Space Operators, Variational Methods, and Applications in Quantum Physics}}}},\ Vol.~\bibinfo {volume} {26}\ (\bibinfo  {publisher} {Birkh{\"a}user},\ \bibinfo {year} {2003})\BibitemShut {NoStop}%
\bibitem [{\citenamefont {Karuvade}(2022)}]{Kar22}%
  \BibitemOpen
  \bibfield  {author} {\bibinfo {author} {\bibfnamefont {S.}~\bibnamefont {Karuvade}},\ }\emph {\bibinfo {title} {Power and Certifiability of Quantum Computing for Open Systems}},\ \href {http://hdl.handle.net/1880/116407} {Ph.D. thesis},\ \bibinfo  {school} {University of Calgary}, \bibinfo {address} {Calgary, Alberta} (\bibinfo {year} {2022})\BibitemShut {NoStop}%
\bibitem [{\citenamefont {Mostafazadeh}(2002{\natexlab{a}})}]{BiOrthogonal}%
  \BibitemOpen
  \bibfield  {author} {\bibinfo {author} {\bibfnamefont {A.}~\bibnamefont {Mostafazadeh}},\ }\bibfield  {title} {\bibinfo {title} {{Pseudo-Hermiticity versus {PT} symmetry: The necessary condition for the reality of the spectrum of a non-Hermitian Hamiltonian}},\ }\href {https://doi.org/10.1063/1.1418246} {\bibfield  {journal} {\bibinfo  {journal} {J. Math. Phys.}\ }\textbf {\bibinfo {volume} {43}},\ \bibinfo {pages} {205} (\bibinfo {year} {2002}{\natexlab{a}})},\ \Eprint {https://arxiv.org/abs/math-ph/0107001} {arXiv:math-ph/0107001} \BibitemShut {NoStop}%
\bibitem [{\citenamefont {Mostafazadeh}(2010{\natexlab{a}})}]{Mos10b}%
  \BibitemOpen
  \bibfield  {author} {\bibinfo {author} {\bibfnamefont {A.}~\bibnamefont {Mostafazadeh}},\ }\bibfield  {title} {\bibinfo {title} {Pseudo-{Hermitian} representation of quantum mechanics},\ }\href@noop {} {\bibfield  {journal} {\bibinfo  {journal} {Int. J. Geom. Meth. Mod. Phys.}\ }\textbf {\bibinfo {volume} {7}},\ \bibinfo {pages} {1191} (\bibinfo {year} {2010}{\natexlab{a}})}\BibitemShut {NoStop}%
\bibitem [{\citenamefont {Sakurai}\ and\ \citenamefont {Napolitano}(2007)}]{SN21}%
  \BibitemOpen
  \bibfield  {author} {\bibinfo {author} {\bibfnamefont {J.}~\bibnamefont {Sakurai}}\ and\ \bibinfo {author} {\bibfnamefont {J.}~\bibnamefont {Napolitano}},\ }\href@noop {} {\emph {\bibinfo {title} {{Modern Quantum Mechanics}}}},\ \bibinfo {edition} {3rd}\ ed.\ (\bibinfo  {publisher} {Cambridge University Press},\ \bibinfo {address} {Cambridge},\ \bibinfo {year} {2007})\BibitemShut {NoStop}%
\bibitem [{\citenamefont {Bender}\ and\ \citenamefont {Boettcher}(1998)}]{BB98}%
  \BibitemOpen
  \bibfield  {author} {\bibinfo {author} {\bibfnamefont {C.~M.}\ \bibnamefont {Bender}}\ and\ \bibinfo {author} {\bibfnamefont {S.}~\bibnamefont {Boettcher}},\ }\bibfield  {title} {\bibinfo {title} {Real spectra in {non-Hermitian Hamiltonians} having $\mathscr{P}\mathscr{T}$ symmetry},\ }\href@noop {} {\bibfield  {journal} {\bibinfo  {journal} {Phys. Rev. Lett.}\ }\textbf {\bibinfo {volume} {80}},\ \bibinfo {pages} {5243} (\bibinfo {year} {1998})}\BibitemShut {NoStop}%
\bibitem [{\citenamefont {Bender}\ \emph {et~al.}(2002)\citenamefont {Bender}, \citenamefont {Brody},\ and\ \citenamefont {Jones}}]{BBJ02}%
  \BibitemOpen
  \bibfield  {author} {\bibinfo {author} {\bibfnamefont {C.~M.}\ \bibnamefont {Bender}}, \bibinfo {author} {\bibfnamefont {D.~C.}\ \bibnamefont {Brody}},\ and\ \bibinfo {author} {\bibfnamefont {H.~F.}\ \bibnamefont {Jones}},\ }\bibfield  {title} {\bibinfo {title} {Complex extension of quantum mechanics},\ }\href@noop {} {\bibfield  {journal} {\bibinfo  {journal} {Phys. Rev. Lett.}\ }\textbf {\bibinfo {volume} {89}},\ \bibinfo {pages} {270401} (\bibinfo {year} {2002})}\BibitemShut {NoStop}%
\bibitem [{\citenamefont {Bender}\ \emph {et~al.}(2003)\citenamefont {Bender}, \citenamefont {Brody},\ and\ \citenamefont {Jones}}]{BBJ03}%
  \BibitemOpen
  \bibfield  {author} {\bibinfo {author} {\bibfnamefont {C.~M.}\ \bibnamefont {Bender}}, \bibinfo {author} {\bibfnamefont {D.~C.}\ \bibnamefont {Brody}},\ and\ \bibinfo {author} {\bibfnamefont {H.~F.}\ \bibnamefont {Jones}},\ }\bibfield  {title} {\bibinfo {title} {{Must a Hamiltonian be Hermitian?}},\ }\href@noop {} {\bibfield  {journal} {\bibinfo  {journal} {Am. J. Phys.}\ }\textbf {\bibinfo {volume} {71}},\ \bibinfo {pages} {1095} (\bibinfo {year} {2003})}\BibitemShut {NoStop}%
\bibitem [{\citenamefont {Bender}(2005)}]{Ben05}%
  \BibitemOpen
  \bibfield  {author} {\bibinfo {author} {\bibfnamefont {C.~M.}\ \bibnamefont {Bender}},\ }\bibfield  {title} {\bibinfo {title} {{Introduction to $\mathscr{PT}$-symmetric quantum theory}},\ }\href@noop {} {\bibfield  {journal} {\bibinfo  {journal} {Contemp. Phys.}\ }\textbf {\bibinfo {volume} {46}},\ \bibinfo {pages} {277} (\bibinfo {year} {2005})}\BibitemShut {NoStop}%
\bibitem [{\citenamefont {R{\"u}ter}\ \emph {et~al.}(2010)\citenamefont {R{\"u}ter}, \citenamefont {Makris}, \citenamefont {El-Ganainy}, \citenamefont {Christodoulides}, \citenamefont {Segev},\ and\ \citenamefont {Kip}}]{RME+10}%
  \BibitemOpen
  \bibfield  {author} {\bibinfo {author} {\bibfnamefont {C.~E.}\ \bibnamefont {R{\"u}ter}}, \bibinfo {author} {\bibfnamefont {K.~G.}\ \bibnamefont {Makris}}, \bibinfo {author} {\bibfnamefont {R.}~\bibnamefont {El-Ganainy}}, \bibinfo {author} {\bibfnamefont {D.~N.}\ \bibnamefont {Christodoulides}}, \bibinfo {author} {\bibfnamefont {M.}~\bibnamefont {Segev}},\ and\ \bibinfo {author} {\bibfnamefont {D.}~\bibnamefont {Kip}},\ }\bibfield  {title} {\bibinfo {title} {Observation of parity--time symmetry in optics},\ }\href@noop {} {\bibfield  {journal} {\bibinfo  {journal} {Nat. Phys.}\ }\textbf {\bibinfo {volume} {6}},\ \bibinfo {pages} {192} (\bibinfo {year} {2010})}\BibitemShut {NoStop}%
\bibitem [{\citenamefont {Schindler}\ \emph {et~al.}(2011)\citenamefont {Schindler}, \citenamefont {Li}, \citenamefont {Zheng}, \citenamefont {Ellis},\ and\ \citenamefont {Kottos}}]{SLZ+11}%
  \BibitemOpen
  \bibfield  {author} {\bibinfo {author} {\bibfnamefont {J.}~\bibnamefont {Schindler}}, \bibinfo {author} {\bibfnamefont {A.}~\bibnamefont {Li}}, \bibinfo {author} {\bibfnamefont {M.~C.}\ \bibnamefont {Zheng}}, \bibinfo {author} {\bibfnamefont {F.~M.}\ \bibnamefont {Ellis}},\ and\ \bibinfo {author} {\bibfnamefont {T.}~\bibnamefont {Kottos}},\ }\bibfield  {title} {\bibinfo {title} {Experimental study of active {LRC} circuits with $\mathcal{PT}$ symmetries},\ }\href@noop {} {\bibfield  {journal} {\bibinfo  {journal} {Phys. Rev. A}\ }\textbf {\bibinfo {volume} {84}},\ \bibinfo {pages} {040101} (\bibinfo {year} {2011})}\BibitemShut {NoStop}%
\bibitem [{\citenamefont {Bittner}\ \emph {et~al.}(2012)\citenamefont {Bittner}, \citenamefont {Dietz}, \citenamefont {G\"unther}, \citenamefont {Harney}, \citenamefont {Miski-Oglu}, \citenamefont {Richter},\ and\ \citenamefont {Sch\"afer}}]{BDG+12}%
  \BibitemOpen
  \bibfield  {author} {\bibinfo {author} {\bibfnamefont {S.}~\bibnamefont {Bittner}}, \bibinfo {author} {\bibfnamefont {B.}~\bibnamefont {Dietz}}, \bibinfo {author} {\bibfnamefont {U.}~\bibnamefont {G\"unther}}, \bibinfo {author} {\bibfnamefont {H.~L.}\ \bibnamefont {Harney}}, \bibinfo {author} {\bibfnamefont {M.}~\bibnamefont {Miski-Oglu}}, \bibinfo {author} {\bibfnamefont {A.}~\bibnamefont {Richter}},\ and\ \bibinfo {author} {\bibfnamefont {F.}~\bibnamefont {Sch\"afer}},\ }\bibfield  {title} {\bibinfo {title} {$\mathscr{P}\mathscr{T}$ symmetry and spontaneous symmetry breaking in a microwave billiard},\ }\href@noop {} {\bibfield  {journal} {\bibinfo  {journal} {Phys. Rev. Lett.}\ }\textbf {\bibinfo {volume} {108}},\ \bibinfo {pages} {024101} (\bibinfo {year} {2012})}\BibitemShut {NoStop}%
\bibitem [{\citenamefont {Peng}\ \emph {et~al.}(2014)\citenamefont {Peng}, \citenamefont {{\"O}zdemir}, \citenamefont {Lei}, \citenamefont {Monifi}, \citenamefont {Gianfreda}, \citenamefont {Long}, \citenamefont {Fan}, \citenamefont {Nori}, \citenamefont {Bender},\ and\ \citenamefont {Yang}}]{POS+14}%
  \BibitemOpen
  \bibfield  {author} {\bibinfo {author} {\bibfnamefont {B.}~\bibnamefont {Peng}}, \bibinfo {author} {\bibfnamefont {{\c{S}}.~K.}\ \bibnamefont {{\"O}zdemir}}, \bibinfo {author} {\bibfnamefont {F.}~\bibnamefont {Lei}}, \bibinfo {author} {\bibfnamefont {F.}~\bibnamefont {Monifi}}, \bibinfo {author} {\bibfnamefont {M.}~\bibnamefont {Gianfreda}}, \bibinfo {author} {\bibfnamefont {G.~L.}\ \bibnamefont {Long}}, \bibinfo {author} {\bibfnamefont {S.}~\bibnamefont {Fan}}, \bibinfo {author} {\bibfnamefont {F.}~\bibnamefont {Nori}}, \bibinfo {author} {\bibfnamefont {C.~M.}\ \bibnamefont {Bender}},\ and\ \bibinfo {author} {\bibfnamefont {L.}~\bibnamefont {Yang}},\ }\bibfield  {title} {\bibinfo {title} {Parity--time-symmetric whispering-gallery microcavities},\ }\href@noop {} {\bibfield  {journal} {\bibinfo  {journal} {Nat. Phys.}\ }\textbf {\bibinfo {volume} {10}},\ \bibinfo {pages} {394} (\bibinfo {year} {2014})}\BibitemShut {NoStop}%
\bibitem [{\citenamefont {Zhang}\ \emph {et~al.}(2016)\citenamefont {Zhang}, \citenamefont {Zhang}, \citenamefont {Sheng}, \citenamefont {Yang}, \citenamefont {Miri}, \citenamefont {Christodoulides}, \citenamefont {He}, \citenamefont {Zhang},\ and\ \citenamefont {Xiao}}]{ZZS+16}%
  \BibitemOpen
  \bibfield  {author} {\bibinfo {author} {\bibfnamefont {Z.}~\bibnamefont {Zhang}}, \bibinfo {author} {\bibfnamefont {Y.}~\bibnamefont {Zhang}}, \bibinfo {author} {\bibfnamefont {J.}~\bibnamefont {Sheng}}, \bibinfo {author} {\bibfnamefont {L.}~\bibnamefont {Yang}}, \bibinfo {author} {\bibfnamefont {M.-A.}\ \bibnamefont {Miri}}, \bibinfo {author} {\bibfnamefont {D.~N.}\ \bibnamefont {Christodoulides}}, \bibinfo {author} {\bibfnamefont {B.}~\bibnamefont {He}}, \bibinfo {author} {\bibfnamefont {Y.}~\bibnamefont {Zhang}},\ and\ \bibinfo {author} {\bibfnamefont {M.}~\bibnamefont {Xiao}},\ }\bibfield  {title} {\bibinfo {title} {Observation of parity-time symmetry in optically induced atomic lattices},\ }\href@noop {} {\bibfield  {journal} {\bibinfo  {journal} {Phys. Rev. Lett.}\ }\textbf {\bibinfo {volume} {117}},\ \bibinfo {pages} {123601} (\bibinfo {year} {2016})}\BibitemShut {NoStop}%
\bibitem [{\citenamefont {Xiao}\ \emph {et~al.}(2017)\citenamefont {Xiao}, \citenamefont {Zhan}, \citenamefont {Bian}, \citenamefont {Wang}, \citenamefont {Zhang}, \citenamefont {Wang}, \citenamefont {Li}, \citenamefont {Mochizuki}, \citenamefont {Kim}, \citenamefont {Kawakami} \emph {et~al.}}]{XZB+17}%
  \BibitemOpen
  \bibfield  {author} {\bibinfo {author} {\bibfnamefont {L.}~\bibnamefont {Xiao}}, \bibinfo {author} {\bibfnamefont {X.}~\bibnamefont {Zhan}}, \bibinfo {author} {\bibfnamefont {Z.}~\bibnamefont {Bian}}, \bibinfo {author} {\bibfnamefont {K.}~\bibnamefont {Wang}}, \bibinfo {author} {\bibfnamefont {X.}~\bibnamefont {Zhang}}, \bibinfo {author} {\bibfnamefont {X.}~\bibnamefont {Wang}}, \bibinfo {author} {\bibfnamefont {J.}~\bibnamefont {Li}}, \bibinfo {author} {\bibfnamefont {K.}~\bibnamefont {Mochizuki}}, \bibinfo {author} {\bibfnamefont {D.}~\bibnamefont {Kim}}, \bibinfo {author} {\bibfnamefont {N.}~\bibnamefont {Kawakami}}, \emph {et~al.},\ }\bibfield  {title} {\bibinfo {title} {Observation of topological edge states in parity--time-symmetric quantum walks},\ }\href@noop {} {\bibfield  {journal} {\bibinfo  {journal} {Nat. Phys.}\ }\textbf {\bibinfo {volume} {13}},\ \bibinfo {pages} {1117} (\bibinfo {year} {2017})}\BibitemShut {NoStop}%
\bibitem [{\citenamefont {El-Ganainy}\ \emph {et~al.}(2018)\citenamefont {El-Ganainy}, \citenamefont {Makris}, \citenamefont {Khajavikhan}, \citenamefont {Musslimani}, \citenamefont {Rotter},\ and\ \citenamefont {Christodoulides}}]{EMK+18}%
  \BibitemOpen
  \bibfield  {author} {\bibinfo {author} {\bibfnamefont {R.}~\bibnamefont {El-Ganainy}}, \bibinfo {author} {\bibfnamefont {K.~G.}\ \bibnamefont {Makris}}, \bibinfo {author} {\bibfnamefont {M.}~\bibnamefont {Khajavikhan}}, \bibinfo {author} {\bibfnamefont {Z.~H.}\ \bibnamefont {Musslimani}}, \bibinfo {author} {\bibfnamefont {S.}~\bibnamefont {Rotter}},\ and\ \bibinfo {author} {\bibfnamefont {D.~N.}\ \bibnamefont {Christodoulides}},\ }\bibfield  {title} {\bibinfo {title} {{Non-Hermitian physics and PT symmetry}},\ }\href@noop {} {\bibfield  {journal} {\bibinfo  {journal} {Nat. Phys.}\ }\textbf {\bibinfo {volume} {14}},\ \bibinfo {pages} {11} (\bibinfo {year} {2018})}\BibitemShut {NoStop}%
\bibitem [{\citenamefont {Wu}\ \emph {et~al.}(2019)\citenamefont {Wu}, \citenamefont {Liu}, \citenamefont {Geng}, \citenamefont {Song}, \citenamefont {Ye}, \citenamefont {Duan}, \citenamefont {Rong},\ and\ \citenamefont {Du}}]{WLG+19}%
  \BibitemOpen
  \bibfield  {author} {\bibinfo {author} {\bibfnamefont {Y.}~\bibnamefont {Wu}}, \bibinfo {author} {\bibfnamefont {W.}~\bibnamefont {Liu}}, \bibinfo {author} {\bibfnamefont {J.}~\bibnamefont {Geng}}, \bibinfo {author} {\bibfnamefont {X.}~\bibnamefont {Song}}, \bibinfo {author} {\bibfnamefont {X.}~\bibnamefont {Ye}}, \bibinfo {author} {\bibfnamefont {C.-K.}\ \bibnamefont {Duan}}, \bibinfo {author} {\bibfnamefont {X.}~\bibnamefont {Rong}},\ and\ \bibinfo {author} {\bibfnamefont {J.}~\bibnamefont {Du}},\ }\bibfield  {title} {\bibinfo {title} {Observation of parity-time symmetry breaking in a single-spin system},\ }\href@noop {} {\bibfield  {journal} {\bibinfo  {journal} {Science}\ }\textbf {\bibinfo {volume} {364}},\ \bibinfo {pages} {878} (\bibinfo {year} {2019})}\BibitemShut {NoStop}%
\bibitem [{\citenamefont {Zhang}\ \emph {et~al.}(2020)\citenamefont {Zhang}, \citenamefont {Li}, \citenamefont {Wang}, \citenamefont {Feng}, \citenamefont {Guan},\ and\ \citenamefont {Yao}}]{ZLW+20}%
  \BibitemOpen
  \bibfield  {author} {\bibinfo {author} {\bibfnamefont {J.}~\bibnamefont {Zhang}}, \bibinfo {author} {\bibfnamefont {L.}~\bibnamefont {Li}}, \bibinfo {author} {\bibfnamefont {G.}~\bibnamefont {Wang}}, \bibinfo {author} {\bibfnamefont {X.}~\bibnamefont {Feng}}, \bibinfo {author} {\bibfnamefont {B.-O.}\ \bibnamefont {Guan}},\ and\ \bibinfo {author} {\bibfnamefont {J.}~\bibnamefont {Yao}},\ }\bibfield  {title} {\bibinfo {title} {Parity-time symmetry in wavelength space within a single spatial resonator},\ }\href@noop {} {\bibfield  {journal} {\bibinfo  {journal} {Nat. Commun.}\ }\textbf {\bibinfo {volume} {11}},\ \bibinfo {pages} {1} (\bibinfo {year} {2020})}\BibitemShut {NoStop}%
\bibitem [{\citenamefont {Mostafazadeh}(2003)}]{Mos03a}%
  \BibitemOpen
  \bibfield  {author} {\bibinfo {author} {\bibfnamefont {A.}~\bibnamefont {Mostafazadeh}},\ }\bibfield  {title} {\bibinfo {title} {{Exact $PT$-symmetry is equivalent to Hermiticity}},\ }\href@noop {} {\bibfield  {journal} {\bibinfo  {journal} {J. Phys. A Math. Theor.}\ }\textbf {\bibinfo {volume} {36}},\ \bibinfo {pages} {7081} (\bibinfo {year} {2003})}\BibitemShut {NoStop}%
\bibitem [{\citenamefont {Makris}\ \emph {et~al.}(2008)\citenamefont {Makris}, \citenamefont {El-Ganainy}, \citenamefont {Christodoulides},\ and\ \citenamefont {Musslimani}}]{MECM08}%
  \BibitemOpen
  \bibfield  {author} {\bibinfo {author} {\bibfnamefont {K.~G.}\ \bibnamefont {Makris}}, \bibinfo {author} {\bibfnamefont {R.}~\bibnamefont {El-Ganainy}}, \bibinfo {author} {\bibfnamefont {D.~N.}\ \bibnamefont {Christodoulides}},\ and\ \bibinfo {author} {\bibfnamefont {Z.~H.}\ \bibnamefont {Musslimani}},\ }\bibfield  {title} {\bibinfo {title} {Beam dynamics in $\mathcal{P}\mathcal{T}$ symmetric optical lattices},\ }\href {https://doi.org/10.1103/PhysRevLett.100.103904} {\bibfield  {journal} {\bibinfo  {journal} {Phys. Rev. Lett.}\ }\textbf {\bibinfo {volume} {100}},\ \bibinfo {pages} {103904} (\bibinfo {year} {2008})}\BibitemShut {NoStop}%
\bibitem [{\citenamefont {Ju}\ \emph {et~al.}(2019)\citenamefont {Ju}, \citenamefont {Miranowicz}, \citenamefont {Chen},\ and\ \citenamefont {Nori}}]{JMCN19}%
  \BibitemOpen
  \bibfield  {author} {\bibinfo {author} {\bibfnamefont {C.-Y.}\ \bibnamefont {Ju}}, \bibinfo {author} {\bibfnamefont {A.}~\bibnamefont {Miranowicz}}, \bibinfo {author} {\bibfnamefont {G.-Y.}\ \bibnamefont {Chen}},\ and\ \bibinfo {author} {\bibfnamefont {F.}~\bibnamefont {Nori}},\ }\bibfield  {title} {\bibinfo {title} {{Non-Hermitian Hamiltonians and no-go theorems in quantum information}},\ }\href@noop {} {\bibfield  {journal} {\bibinfo  {journal} {Phys. Rev. A}\ }\textbf {\bibinfo {volume} {100}},\ \bibinfo {pages} {062118} (\bibinfo {year} {2019})}\BibitemShut {NoStop}%
\bibitem [{\citenamefont {Croke}(2015)}]{Cro15}%
  \BibitemOpen
  \bibfield  {author} {\bibinfo {author} {\bibfnamefont {S.}~\bibnamefont {Croke}},\ }\bibfield  {title} {\bibinfo {title} {{$\mathcal{PT}$-symmetric Hamiltonians and their application in quantum information}},\ }\href@noop {} {\bibfield  {journal} {\bibinfo  {journal} {Phys. Rev. A}\ }\textbf {\bibinfo {volume} {91}},\ \bibinfo {pages} {052113} (\bibinfo {year} {2015})}\BibitemShut {NoStop}%
\bibitem [{\citenamefont {Bender}\ \emph {et~al.}(2013)\citenamefont {Bender}, \citenamefont {Brody}, \citenamefont {Caldeira}, \citenamefont {G{\"u}nther}, \citenamefont {Meister},\ and\ \citenamefont {Samsonov}}]{BBC+13}%
  \BibitemOpen
  \bibfield  {author} {\bibinfo {author} {\bibfnamefont {C.~M.}\ \bibnamefont {Bender}}, \bibinfo {author} {\bibfnamefont {D.~C.}\ \bibnamefont {Brody}}, \bibinfo {author} {\bibfnamefont {J.}~\bibnamefont {Caldeira}}, \bibinfo {author} {\bibfnamefont {U.}~\bibnamefont {G{\"u}nther}}, \bibinfo {author} {\bibfnamefont {B.~K.}\ \bibnamefont {Meister}},\ and\ \bibinfo {author} {\bibfnamefont {B.~F.}\ \bibnamefont {Samsonov}},\ }\bibfield  {title} {\bibinfo {title} {{PT}-symmetric quantum state discrimination},\ }\href@noop {} {\bibfield  {journal} {\bibinfo  {journal} {Philos. Trans. R. Soc. A}\ }\textbf {\bibinfo {volume} {371}},\ \bibinfo {pages} {20120160} (\bibinfo {year} {2013})}\BibitemShut {NoStop}%
\bibitem [{\citenamefont {Zhan}\ \emph {et~al.}(2020)\citenamefont {Zhan}, \citenamefont {Wang}, \citenamefont {Xiao}, \citenamefont {Bian}, \citenamefont {Zhang}, \citenamefont {Sanders}, \citenamefont {Zhang},\ and\ \citenamefont {Xue}}]{ZWX+20}%
  \BibitemOpen
  \bibfield  {author} {\bibinfo {author} {\bibfnamefont {X.}~\bibnamefont {Zhan}}, \bibinfo {author} {\bibfnamefont {K.}~\bibnamefont {Wang}}, \bibinfo {author} {\bibfnamefont {L.}~\bibnamefont {Xiao}}, \bibinfo {author} {\bibfnamefont {Z.}~\bibnamefont {Bian}}, \bibinfo {author} {\bibfnamefont {Y.}~\bibnamefont {Zhang}}, \bibinfo {author} {\bibfnamefont {B.~C.}\ \bibnamefont {Sanders}}, \bibinfo {author} {\bibfnamefont {C.}~\bibnamefont {Zhang}},\ and\ \bibinfo {author} {\bibfnamefont {P.}~\bibnamefont {Xue}},\ }\bibfield  {title} {\bibinfo {title} {Experimental quantum cloning in a pseudo-unitary system},\ }\href@noop {} {\bibfield  {journal} {\bibinfo  {journal} {Phys. Rev. A}\ }\textbf {\bibinfo {volume} {101}},\ \bibinfo {pages} {010302} (\bibinfo {year} {2020})}\BibitemShut {NoStop}%
\bibitem [{\citenamefont {Bender}\ \emph {et~al.}(2007)\citenamefont {Bender}, \citenamefont {Brody}, \citenamefont {Jones},\ and\ \citenamefont {Meister}}]{BBJM07}%
  \BibitemOpen
  \bibfield  {author} {\bibinfo {author} {\bibfnamefont {C.~M.}\ \bibnamefont {Bender}}, \bibinfo {author} {\bibfnamefont {D.~C.}\ \bibnamefont {Brody}}, \bibinfo {author} {\bibfnamefont {H.~F.}\ \bibnamefont {Jones}},\ and\ \bibinfo {author} {\bibfnamefont {B.~K.}\ \bibnamefont {Meister}},\ }\bibfield  {title} {\bibinfo {title} {Faster than {Hermitian} quantum mechanics},\ }\href@noop {} {\bibfield  {journal} {\bibinfo  {journal} {Phys. Rev. Lett.}\ }\textbf {\bibinfo {volume} {98}},\ \bibinfo {pages} {040403} (\bibinfo {year} {2007})}\BibitemShut {NoStop}%
\bibitem [{\citenamefont {Mostafazadeh}(2009)}]{Mos09}%
  \BibitemOpen
  \bibfield  {author} {\bibinfo {author} {\bibfnamefont {A.}~\bibnamefont {Mostafazadeh}},\ }\bibfield  {title} {\bibinfo {title} {Hamiltonians generating optimal-speed evolutions},\ }\href@noop {} {\bibfield  {journal} {\bibinfo  {journal} {Phys. Rev. A}\ }\textbf {\bibinfo {volume} {79}},\ \bibinfo {pages} {014101} (\bibinfo {year} {2009})}\BibitemShut {NoStop}%
\bibitem [{\citenamefont {Pati}()}]{Pat14}%
  \BibitemOpen
  \bibfield  {author} {\bibinfo {author} {\bibfnamefont {A.~K.}\ \bibnamefont {Pati}},\ }\bibfield  {title} {\bibinfo {title} {{Violation of invariance of entanglement under local PT symmetric unitary}},\ }\href@noop {} {\bibinfo  {journal} {arXiv:1404.6166}\ }\BibitemShut {NoStop}%
\bibitem [{\citenamefont {Chen}\ \emph {et~al.}(2014)\citenamefont {Chen}, \citenamefont {Chen},\ and\ \citenamefont {Chen}}]{CCC14}%
  \BibitemOpen
\bibfield  {journal} {  }\bibfield  {author} {\bibinfo {author} {\bibfnamefont {S.-L.}\ \bibnamefont {Chen}}, \bibinfo {author} {\bibfnamefont {G.-Y.}\ \bibnamefont {Chen}},\ and\ \bibinfo {author} {\bibfnamefont {Y.-N.}\ \bibnamefont {Chen}},\ }\bibfield  {title} {\bibinfo {title} {Increase of entanglement by local $\mathcal{PT}$-symmetric operations},\ }\href@noop {} {\bibfield  {journal} {\bibinfo  {journal} {Phys. Rev. A}\ }\textbf {\bibinfo {volume} {90}},\ \bibinfo {pages} {054301} (\bibinfo {year} {2014})}\BibitemShut {NoStop}%
\bibitem [{\citenamefont {Lee}\ \emph {et~al.}(2014)\citenamefont {Lee}, \citenamefont {Hsieh}, \citenamefont {Flammia},\ and\ \citenamefont {Lee}}]{YHFL14}%
  \BibitemOpen
  \bibfield  {author} {\bibinfo {author} {\bibfnamefont {Y.-C.}\ \bibnamefont {Lee}}, \bibinfo {author} {\bibfnamefont {M.-H.}\ \bibnamefont {Hsieh}}, \bibinfo {author} {\bibfnamefont {S.~T.}\ \bibnamefont {Flammia}},\ and\ \bibinfo {author} {\bibfnamefont {R.-K.}\ \bibnamefont {Lee}},\ }\bibfield  {title} {\bibinfo {title} {Local {$\mathcal{P}\mathcal{T}$} symmetry violates the no-signaling principle},\ }\href@noop {} {\bibfield  {journal} {\bibinfo  {journal} {Phys. Rev. Lett.}\ }\textbf {\bibinfo {volume} {112}},\ \bibinfo {pages} {130404} (\bibinfo {year} {2014})}\BibitemShut {NoStop}%
\bibitem [{\citenamefont {Mostafazadeh}(2010{\natexlab{b}})}]{Mos10a}%
  \BibitemOpen
  \bibfield  {author} {\bibinfo {author} {\bibfnamefont {A.}~\bibnamefont {Mostafazadeh}},\ }\bibfield  {title} {\bibinfo {title} {Conceptual aspects of $\mathcal{P}\mathcal{T}$-symmetry and pseudo-{Hermiticity}: a status report},\ }\href@noop {} {\bibfield  {journal} {\bibinfo  {journal} {Phys. Scr.}\ }\textbf {\bibinfo {volume} {82}},\ \bibinfo {pages} {038110} (\bibinfo {year} {2010}{\natexlab{b}})}\BibitemShut {NoStop}%
\bibitem [{\citenamefont {Znojil}(2015)}]{Zno15}%
  \BibitemOpen
  \bibfield  {author} {\bibinfo {author} {\bibfnamefont {M.}~\bibnamefont {Znojil}},\ }\href@noop {} {\emph {\bibinfo {title} {{Non-self-adjoint Operators in Quantum Physics: Ideas, People, and Trends}}}},\ edited by\ \bibinfo {editor} {\bibfnamefont {F.}~\bibnamefont {Bagarello}}, \bibinfo {editor} {\bibfnamefont {J.~P.}\ \bibnamefont {Gazeau}}, \bibinfo {editor} {\bibfnamefont {F.~H.}\ \bibnamefont {Szafraniec}},\ and\ \bibinfo {editor} {\bibfnamefont {M.}~\bibnamefont {Znojil}}\ (\bibinfo  {publisher} {Wiley},\ \bibinfo {year} {2015})\BibitemShut {NoStop}%
\bibitem [{\citenamefont {Mostafazadeh}(2018)}]{Mos18}%
  \BibitemOpen
  \bibfield  {author} {\bibinfo {author} {\bibfnamefont {A.}~\bibnamefont {Mostafazadeh}},\ }\bibfield  {title} {\bibinfo {title} {{Energy observable for a quantum system with a dynamical {Hilbert} space and a global geometric extension of quantum theory}},\ }\href@noop {} {\bibfield  {journal} {\bibinfo  {journal} {Phys. Rev. D}\ }\textbf {\bibinfo {volume} {98}},\ \bibinfo {pages} {046022} (\bibinfo {year} {2018})}\BibitemShut {NoStop}%
\bibitem [{\citenamefont {Zhang}\ \emph {et~al.}(2019)\citenamefont {Zhang}, \citenamefont {Wang},\ and\ \citenamefont {Gong}}]{ZWG19}%
  \BibitemOpen
  \bibfield  {author} {\bibinfo {author} {\bibfnamefont {D.-J.}\ \bibnamefont {Zhang}}, \bibinfo {author} {\bibfnamefont {Q.-h.}\ \bibnamefont {Wang}},\ and\ \bibinfo {author} {\bibfnamefont {J.}~\bibnamefont {Gong}},\ }\bibfield  {title} {\bibinfo {title} {Time-dependent $\mathscr{PT}$-symmetric quantum mechanics in generic non-hermitian systems},\ }\href@noop {} {\bibfield  {journal} {\bibinfo  {journal} {Phys. Rev. A}\ }\textbf {\bibinfo {volume} {100}},\ \bibinfo {pages} {062121} (\bibinfo {year} {2019})}\BibitemShut {NoStop}%
\bibitem [{\citenamefont {Brody}(2016)}]{Bro16}%
  \BibitemOpen
  \bibfield  {author} {\bibinfo {author} {\bibfnamefont {D.~C.}\ \bibnamefont {Brody}},\ }\bibfield  {title} {\bibinfo {title} {Consistency of {PT}-symmetric quantum mechanics},\ }\href@noop {} {\bibfield  {journal} {\bibinfo  {journal} {J. Phys. A Math. Theor.}\ }\textbf {\bibinfo {volume} {49}},\ \bibinfo {pages} {10LT03} (\bibinfo {year} {2016})}\BibitemShut {NoStop}%
\bibitem [{\citenamefont {Liu}\ \emph {et~al.}(2016)\citenamefont {Liu}, \citenamefont {Zhang}, \citenamefont {\"Ozdemir}, \citenamefont {Peng}, \citenamefont {Jing}, \citenamefont {L\"u}, \citenamefont {Li}, \citenamefont {Yang}, \citenamefont {Nori},\ and\ \citenamefont {Liu}}]{LZO+16}%
  \BibitemOpen
  \bibfield  {author} {\bibinfo {author} {\bibfnamefont {Z.-P.}\ \bibnamefont {Liu}}, \bibinfo {author} {\bibfnamefont {J.}~\bibnamefont {Zhang}}, \bibinfo {author} {\bibfnamefont {{\c{S}}.~K.}\ \bibnamefont {\"Ozdemir}}, \bibinfo {author} {\bibfnamefont {B.}~\bibnamefont {Peng}}, \bibinfo {author} {\bibfnamefont {H.}~\bibnamefont {Jing}}, \bibinfo {author} {\bibfnamefont {X.-Y.}\ \bibnamefont {L\"u}}, \bibinfo {author} {\bibfnamefont {C.-W.}\ \bibnamefont {Li}}, \bibinfo {author} {\bibfnamefont {L.}~\bibnamefont {Yang}}, \bibinfo {author} {\bibfnamefont {F.}~\bibnamefont {Nori}},\ and\ \bibinfo {author} {\bibfnamefont {Y.-x.}\ \bibnamefont {Liu}},\ }\bibfield  {title} {\bibinfo {title} {Metrology with $\mathcal{PT}$-symmetric cavities: {Enhanced} sensitivity near the $\mathcal{PT}$-phase transition},\ }\href@noop {} {\bibfield  {journal} {\bibinfo  {journal} {Phys. Rev. Lett.}\ }\textbf {\bibinfo {volume} {117}},\ \bibinfo {pages} {110802} (\bibinfo {year} {2016})}\BibitemShut {NoStop}%
\bibitem [{\citenamefont {Chen}\ \emph {et~al.}(2017)\citenamefont {Chen}, \citenamefont {{\"O}zdemir}, \citenamefont {Zhao}, \citenamefont {Wiersig},\ and\ \citenamefont {Yang}}]{COZ+17}%
  \BibitemOpen
  \bibfield  {author} {\bibinfo {author} {\bibfnamefont {W.}~\bibnamefont {Chen}}, \bibinfo {author} {\bibfnamefont {{\c{S}}.~K.}\ \bibnamefont {{\"O}zdemir}}, \bibinfo {author} {\bibfnamefont {G.}~\bibnamefont {Zhao}}, \bibinfo {author} {\bibfnamefont {J.}~\bibnamefont {Wiersig}},\ and\ \bibinfo {author} {\bibfnamefont {L.}~\bibnamefont {Yang}},\ }\bibfield  {title} {\bibinfo {title} {Exceptional points enhance sensing in an optical microcavity},\ }\href@noop {} {\bibfield  {journal} {\bibinfo  {journal} {Nature}\ }\textbf {\bibinfo {volume} {548}},\ \bibinfo {pages} {192} (\bibinfo {year} {2017})}\BibitemShut {NoStop}%
\bibitem [{\citenamefont {Hodaei}\ \emph {et~al.}(2017)\citenamefont {Hodaei}, \citenamefont {Hassan}, \citenamefont {Wittek}, \citenamefont {Garcia-Gracia}, \citenamefont {El-Ganainy}, \citenamefont {Christodoulides},\ and\ \citenamefont {Khajavikhan}}]{HHW+17}%
  \BibitemOpen
  \bibfield  {author} {\bibinfo {author} {\bibfnamefont {H.}~\bibnamefont {Hodaei}}, \bibinfo {author} {\bibfnamefont {A.~U.}\ \bibnamefont {Hassan}}, \bibinfo {author} {\bibfnamefont {S.}~\bibnamefont {Wittek}}, \bibinfo {author} {\bibfnamefont {H.}~\bibnamefont {Garcia-Gracia}}, \bibinfo {author} {\bibfnamefont {R.}~\bibnamefont {El-Ganainy}}, \bibinfo {author} {\bibfnamefont {D.~N.}\ \bibnamefont {Christodoulides}},\ and\ \bibinfo {author} {\bibfnamefont {M.}~\bibnamefont {Khajavikhan}},\ }\bibfield  {title} {\bibinfo {title} {Enhanced sensitivity at higher-order exceptional points},\ }\href@noop {} {\bibfield  {journal} {\bibinfo  {journal} {Nature}\ }\textbf {\bibinfo {volume} {548}},\ \bibinfo {pages} {187} (\bibinfo {year} {2017})}\BibitemShut {NoStop}%
\bibitem [{\citenamefont {Zhu}\ \emph {et~al.}(2013)\citenamefont {Zhu}, \citenamefont {Feng}, \citenamefont {Zhang}, \citenamefont {Yin},\ and\ \citenamefont {Zhang}}]{ZFZ+13}%
  \BibitemOpen
  \bibfield  {author} {\bibinfo {author} {\bibfnamefont {X.}~\bibnamefont {Zhu}}, \bibinfo {author} {\bibfnamefont {L.}~\bibnamefont {Feng}}, \bibinfo {author} {\bibfnamefont {P.}~\bibnamefont {Zhang}}, \bibinfo {author} {\bibfnamefont {X.}~\bibnamefont {Yin}},\ and\ \bibinfo {author} {\bibfnamefont {X.}~\bibnamefont {Zhang}},\ }\bibfield  {title} {\bibinfo {title} {One-way invisible cloak using parity-time symmetric transformation optics},\ }\href@noop {} {\bibfield  {journal} {\bibinfo  {journal} {Opt. Lett.}\ }\textbf {\bibinfo {volume} {38}},\ \bibinfo {pages} {2821} (\bibinfo {year} {2013})}\BibitemShut {NoStop}%
\bibitem [{\citenamefont {Sounas}\ \emph {et~al.}(2015)\citenamefont {Sounas}, \citenamefont {Fleury},\ and\ \citenamefont {Al{\`u}}}]{SFA15}%
  \BibitemOpen
  \bibfield  {author} {\bibinfo {author} {\bibfnamefont {D.~L.}\ \bibnamefont {Sounas}}, \bibinfo {author} {\bibfnamefont {R.}~\bibnamefont {Fleury}},\ and\ \bibinfo {author} {\bibfnamefont {A.}~\bibnamefont {Al{\`u}}},\ }\bibfield  {title} {\bibinfo {title} {Unidirectional cloaking based on metasurfaces with balanced loss and gain},\ }\href@noop {} {\bibfield  {journal} {\bibinfo  {journal} {Phys. Rev. Appl.}\ }\textbf {\bibinfo {volume} {4}},\ \bibinfo {pages} {014005} (\bibinfo {year} {2015})}\BibitemShut {NoStop}%
\bibitem [{\citenamefont {Ramezani}\ \emph {et~al.}(2010)\citenamefont {Ramezani}, \citenamefont {Kottos}, \citenamefont {El-Ganainy},\ and\ \citenamefont {Christodoulides}}]{RKEC10}%
  \BibitemOpen
  \bibfield  {author} {\bibinfo {author} {\bibfnamefont {H.}~\bibnamefont {Ramezani}}, \bibinfo {author} {\bibfnamefont {T.}~\bibnamefont {Kottos}}, \bibinfo {author} {\bibfnamefont {R.}~\bibnamefont {El-Ganainy}},\ and\ \bibinfo {author} {\bibfnamefont {D.~N.}\ \bibnamefont {Christodoulides}},\ }\bibfield  {title} {\bibinfo {title} {Unidirectional nonlinear {PT}-symmetric optical structures},\ }\href@noop {} {\bibfield  {journal} {\bibinfo  {journal} {Phys. Rev. A}\ }\textbf {\bibinfo {volume} {82}},\ \bibinfo {pages} {043803} (\bibinfo {year} {2010})}\BibitemShut {NoStop}%
\bibitem [{\citenamefont {Li}\ \emph {et~al.}(2019)\citenamefont {Li}, \citenamefont {Harter}, \citenamefont {Liu}, \citenamefont {de~Melo}, \citenamefont {Joglekar},\ and\ \citenamefont {Luo}}]{LHL+19}%
  \BibitemOpen
  \bibfield  {author} {\bibinfo {author} {\bibfnamefont {J.}~\bibnamefont {Li}}, \bibinfo {author} {\bibfnamefont {A.~K.}\ \bibnamefont {Harter}}, \bibinfo {author} {\bibfnamefont {J.}~\bibnamefont {Liu}}, \bibinfo {author} {\bibfnamefont {L.}~\bibnamefont {de~Melo}}, \bibinfo {author} {\bibfnamefont {Y.~N.}\ \bibnamefont {Joglekar}},\ and\ \bibinfo {author} {\bibfnamefont {L.}~\bibnamefont {Luo}},\ }\bibfield  {title} {\bibinfo {title} {Observation of parity-time symmetry breaking transitions in a dissipative {Floquet} system of ultracold atoms},\ }\href@noop {} {\bibfield  {journal} {\bibinfo  {journal} {Nat. Commun.}\ }\textbf {\bibinfo {volume} {10}},\ \bibinfo {pages} {1} (\bibinfo {year} {2019})}\BibitemShut {NoStop}%
\bibitem [{\citenamefont {Xiao}\ \emph {et~al.}()\citenamefont {Xiao}, \citenamefont {Qiu}, \citenamefont {Wang}, \citenamefont {Sanders}, \citenamefont {Yi},\ and\ \citenamefont {Xue}}]{XQW+19}%
  \BibitemOpen
  \bibfield  {author} {\bibinfo {author} {\bibfnamefont {L.}~\bibnamefont {Xiao}}, \bibinfo {author} {\bibfnamefont {X.}~\bibnamefont {Qiu}}, \bibinfo {author} {\bibfnamefont {K.}~\bibnamefont {Wang}}, \bibinfo {author} {\bibfnamefont {B.~C.}\ \bibnamefont {Sanders}}, \bibinfo {author} {\bibfnamefont {W.}~\bibnamefont {Yi}},\ and\ \bibinfo {author} {\bibfnamefont {P.}~\bibnamefont {Xue}},\ }\bibfield  {title} {\bibinfo {title} {Topology with broken parity-time symmetry},\ }\href@noop {} {\bibinfo  {journal} {arXiv:1906.07468}\ }\BibitemShut {NoStop}%
\bibitem [{\citenamefont {G\"unther}\ and\ \citenamefont {Samsonov}(2008{\natexlab{a}})}]{GS08a}%
  \BibitemOpen
\bibfield  {journal} {  }\bibfield  {author} {\bibinfo {author} {\bibfnamefont {U.}~\bibnamefont {G\"unther}}\ and\ \bibinfo {author} {\bibfnamefont {B.~F.}\ \bibnamefont {Samsonov}},\ }\bibfield  {title} {\bibinfo {title} {$\mathscr{P}\mathscr{T}$-symmetric brachistochrone problem, {Lorentz} boosts, and nonunitary operator equivalence classes},\ }\href {https://doi.org/10.1103/PhysRevA.78.042115} {\bibfield  {journal} {\bibinfo  {journal} {Phys. Rev. A}\ }\textbf {\bibinfo {volume} {78}},\ \bibinfo {pages} {042115} (\bibinfo {year} {2008}{\natexlab{a}})}\BibitemShut {NoStop}%
\bibitem [{\citenamefont {G\"unther}\ and\ \citenamefont {Samsonov}(2008{\natexlab{b}})}]{GS08b}%
  \BibitemOpen
  \bibfield  {author} {\bibinfo {author} {\bibfnamefont {U.}~\bibnamefont {G\"unther}}\ and\ \bibinfo {author} {\bibfnamefont {B.~F.}\ \bibnamefont {Samsonov}},\ }\bibfield  {title} {\bibinfo {title} {Naimark-dilated $\mathscr{P}\mathscr{T}$-symmetric brachistochrone},\ }\href@noop {} {\bibfield  {journal} {\bibinfo  {journal} {Phys. Rev. Lett.}\ }\textbf {\bibinfo {volume} {101}},\ \bibinfo {pages} {230404} (\bibinfo {year} {2008}{\natexlab{b}})}\BibitemShut {NoStop}%
\bibitem [{\citenamefont {Zheng}\ \emph {et~al.}(2013)\citenamefont {Zheng}, \citenamefont {Hao},\ and\ \citenamefont {Long}}]{ZHL13}%
  \BibitemOpen
  \bibfield  {author} {\bibinfo {author} {\bibfnamefont {C.}~\bibnamefont {Zheng}}, \bibinfo {author} {\bibfnamefont {L.}~\bibnamefont {Hao}},\ and\ \bibinfo {author} {\bibfnamefont {G.~L.}\ \bibnamefont {Long}},\ }\bibfield  {title} {\bibinfo {title} {Observation of a fast evolution in a parity-time-symmetric system},\ }\href@noop {} {\bibfield  {journal} {\bibinfo  {journal} {Phil. Trans. R. Soc. A.}\ }\textbf {\bibinfo {volume} {371}},\ \bibinfo {pages} {20120053} (\bibinfo {year} {2013})}\BibitemShut {NoStop}%
\bibitem [{\citenamefont {Tang}\ \emph {et~al.}(2016)\citenamefont {Tang}, \citenamefont {Wang}, \citenamefont {Yu}, \citenamefont {He}, \citenamefont {Xu}, \citenamefont {Liu}, \citenamefont {Chen}, \citenamefont {Sun}, \citenamefont {Sun}, \citenamefont {Han} \emph {et~al.}}]{TWY+16}%
  \BibitemOpen
  \bibfield  {author} {\bibinfo {author} {\bibfnamefont {J.-S.}\ \bibnamefont {Tang}}, \bibinfo {author} {\bibfnamefont {Y.-T.}\ \bibnamefont {Wang}}, \bibinfo {author} {\bibfnamefont {S.}~\bibnamefont {Yu}}, \bibinfo {author} {\bibfnamefont {D.-Y.}\ \bibnamefont {He}}, \bibinfo {author} {\bibfnamefont {J.-S.}\ \bibnamefont {Xu}}, \bibinfo {author} {\bibfnamefont {B.-H.}\ \bibnamefont {Liu}}, \bibinfo {author} {\bibfnamefont {G.}~\bibnamefont {Chen}}, \bibinfo {author} {\bibfnamefont {Y.-N.}\ \bibnamefont {Sun}}, \bibinfo {author} {\bibfnamefont {K.}~\bibnamefont {Sun}}, \bibinfo {author} {\bibfnamefont {Y.-J.}\ \bibnamefont {Han}}, \emph {et~al.},\ }\bibfield  {title} {\bibinfo {title} {Experimental investigation of the no-signalling principle in parity--time symmetric theory using an open quantum system},\ }\href@noop {} {\bibfield  {journal} {\bibinfo  {journal} {Nat. Photonics}\ }\textbf {\bibinfo {volume} {10}},\ \bibinfo {pages} {642} (\bibinfo {year} {2016})}\BibitemShut {NoStop}%
\bibitem [{\citenamefont {Kawabata}\ \emph {et~al.}(2017)\citenamefont {Kawabata}, \citenamefont {Ashida},\ and\ \citenamefont {Ueda}}]{KAU17}%
  \BibitemOpen
  \bibfield  {author} {\bibinfo {author} {\bibfnamefont {K.}~\bibnamefont {Kawabata}}, \bibinfo {author} {\bibfnamefont {Y.}~\bibnamefont {Ashida}},\ and\ \bibinfo {author} {\bibfnamefont {M.}~\bibnamefont {Ueda}},\ }\bibfield  {title} {\bibinfo {title} {Information retrieval and criticality in parity-time-symmetric systems},\ }\href@noop {} {\bibfield  {journal} {\bibinfo  {journal} {Phys. Rev. Lett.}\ }\textbf {\bibinfo {volume} {119}},\ \bibinfo {pages} {190401} (\bibinfo {year} {2017})}\BibitemShut {NoStop}%
\bibitem [{\citenamefont {Huang}\ \emph {et~al.}(2018)\citenamefont {Huang}, \citenamefont {Kumar},\ and\ \citenamefont {Wu}}]{HKW18}%
  \BibitemOpen
  \bibfield  {author} {\bibinfo {author} {\bibfnamefont {M.}~\bibnamefont {Huang}}, \bibinfo {author} {\bibfnamefont {A.}~\bibnamefont {Kumar}},\ and\ \bibinfo {author} {\bibfnamefont {J.}~\bibnamefont {Wu}},\ }\bibfield  {title} {\bibinfo {title} {Embedding, simulation and consistency of $\mathscr{P}\mathscr{T}$-symmetric quantum theory},\ }\href@noop {} {\bibfield  {journal} {\bibinfo  {journal} {Phys. Lett. A}\ }\textbf {\bibinfo {volume} {382}},\ \bibinfo {pages} {2578} (\bibinfo {year} {2018})}\BibitemShut {NoStop}%
\bibitem [{\citenamefont {Xiao}\ \emph {et~al.}(2019)\citenamefont {Xiao}, \citenamefont {Wang}, \citenamefont {Zhan}, \citenamefont {Bian}, \citenamefont {Kawabata}, \citenamefont {Ueda}, \citenamefont {Yi},\ and\ \citenamefont {Xue}}]{XWZ+19}%
  \BibitemOpen
  \bibfield  {author} {\bibinfo {author} {\bibfnamefont {L.}~\bibnamefont {Xiao}}, \bibinfo {author} {\bibfnamefont {K.}~\bibnamefont {Wang}}, \bibinfo {author} {\bibfnamefont {X.}~\bibnamefont {Zhan}}, \bibinfo {author} {\bibfnamefont {Z.}~\bibnamefont {Bian}}, \bibinfo {author} {\bibfnamefont {K.}~\bibnamefont {Kawabata}}, \bibinfo {author} {\bibfnamefont {M.}~\bibnamefont {Ueda}}, \bibinfo {author} {\bibfnamefont {W.}~\bibnamefont {Yi}},\ and\ \bibinfo {author} {\bibfnamefont {P.}~\bibnamefont {Xue}},\ }\bibfield  {title} {\bibinfo {title} {Observation of critical phenomena in parity-time-symmetric quantum dynamics},\ }\href@noop {} {\bibfield  {journal} {\bibinfo  {journal} {Phys. Rev. Lett.}\ }\textbf {\bibinfo {volume} {123}},\ \bibinfo {pages} {230401} (\bibinfo {year} {2019})}\BibitemShut {NoStop}%
\bibitem [{\citenamefont {Gao}\ \emph {et~al.}(2021)\citenamefont {Gao}, \citenamefont {Zheng}, \citenamefont {Liu}, \citenamefont {Wang},\ and\ \citenamefont {Wang}}]{GZL+21}%
  \BibitemOpen
  \bibfield  {author} {\bibinfo {author} {\bibfnamefont {W.-C.}\ \bibnamefont {Gao}}, \bibinfo {author} {\bibfnamefont {C.}~\bibnamefont {Zheng}}, \bibinfo {author} {\bibfnamefont {L.}~\bibnamefont {Liu}}, \bibinfo {author} {\bibfnamefont {T.-J.}\ \bibnamefont {Wang}},\ and\ \bibinfo {author} {\bibfnamefont {C.}~\bibnamefont {Wang}},\ }\bibfield  {title} {\bibinfo {title} {Experimental simulation of the parity-time symmetric dynamics using photonic qubits},\ }\href@noop {} {\bibfield  {journal} {\bibinfo  {journal} {Opt. Express}\ }\textbf {\bibinfo {volume} {29}},\ \bibinfo {pages} {517} (\bibinfo {year} {2021})}\BibitemShut {NoStop}%
\bibitem [{\citenamefont {Strocchi}(2008)}]{Str08}%
  \BibitemOpen
  \bibfield  {author} {\bibinfo {author} {\bibfnamefont {F.}~\bibnamefont {Strocchi}},\ }\href@noop {} {\emph {\bibinfo {title} {{An Introduction to the Mathematical Structure of Quantum Mechanics: A Short Course for Mathematicians}}}},\ \bibinfo {edition} {2nd}\ ed.,\ Advanced Series in Mathematical Physics\ (\bibinfo  {publisher} {World Scientific},\ \bibinfo {address} {Singapore},\ \bibinfo {year} {2008})\BibitemShut {NoStop}%
\bibitem [{\citenamefont {Gelfand}\ and\ \citenamefont {Neumark}(1943)}]{GN43}%
  \BibitemOpen
  \bibfield  {author} {\bibinfo {author} {\bibfnamefont {I.}~\bibnamefont {Gelfand}}\ and\ \bibinfo {author} {\bibfnamefont {M.}~\bibnamefont {Neumark}},\ }\bibfield  {title} {\bibinfo {title} {On the imbedding of normed rings into the ring of operators in {Hilbert} space},\ }\href@noop {} {\bibfield  {journal} {\bibinfo  {journal} {Rec. Math. [Mat. Sbornik] N.S.}\ }\textbf {\bibinfo {volume} {12}},\ \bibinfo {pages} {197} (\bibinfo {year} {1943})}\BibitemShut {NoStop}%
\bibitem [{\citenamefont {Segal}(1947)}]{Seg47}%
  \BibitemOpen
  \bibfield  {author} {\bibinfo {author} {\bibfnamefont {I.}~\bibnamefont {Segal}},\ }\bibfield  {title} {\bibinfo {title} {Irreducible representations of operator algebras},\ }\href@noop {} {\bibfield  {journal} {\bibinfo  {journal} {Bull. Amer. Math. Soc.}\ }\textbf {\bibinfo {volume} {53}},\ \bibinfo {pages} {73} (\bibinfo {year} {1947})}\BibitemShut {NoStop}%
\bibitem [{\citenamefont {Gudder}(1979)}]{Gud79}%
  \BibitemOpen
  \bibfield  {author} {\bibinfo {author} {\bibfnamefont {S.~P.}\ \bibnamefont {Gudder}},\ }\bibfield  {title} {\bibinfo {title} {Axiomatic operational quantum mechanics},\ }\href@noop {} {\bibfield  {journal} {\bibinfo  {journal} {Rep. Math. Phys.}\ }\textbf {\bibinfo {volume} {16}},\ \bibinfo {pages} {147} (\bibinfo {year} {1979})}\BibitemShut {NoStop}%
\bibitem [{Note1()}]{Note1}%
  \BibitemOpen
  \bibinfo {note} {We remark that supernormalized functionals are valid states in some other frameworks for describing system evolution under non-Hermitian Hamiltonians~\cite {Mos07,GKN10,UGRM12}}\BibitemShut {NoStop}%
\bibitem [{Note2()}]{Note2}%
  \BibitemOpen
  \bibinfo {note} {The map lift is defined to act only on the normalized and subnormalized vectors in $\protect \mathscr {H}$. The domain of lift in Fig.~\ref {fig:commutativediagram}(c) is shown to be $\protect \mathscr {H}$ for simplicity, but is specified rigorously in Appendix~\ref {sec:liftmap}.}\BibitemShut {Stop}%
\bibitem [{\citenamefont {Conway}(2007)}]{Con07}%
  \BibitemOpen
  \bibfield  {author} {\bibinfo {author} {\bibfnamefont {J.}~\bibnamefont {Conway}},\ }\href@noop {} {\emph {\bibinfo {title} {{A Course in Functional Analysis}}}},\ \bibinfo {edition} {2nd}\ ed.\ (\bibinfo  {publisher} {Springer-Verlag},\ \bibinfo {address} {New York},\ \bibinfo {year} {2007})\BibitemShut {NoStop}%
\bibitem [{\citenamefont {Chuang}\ and\ \citenamefont {Nielsen}(1997)}]{CN97}%
  \BibitemOpen
  \bibfield  {author} {\bibinfo {author} {\bibfnamefont {I.~L.}\ \bibnamefont {Chuang}}\ and\ \bibinfo {author} {\bibfnamefont {M.~A.}\ \bibnamefont {Nielsen}},\ }\bibfield  {title} {\bibinfo {title} {Prescription for experimental determination of the dynamics of a quantum black box},\ }\href@noop {} {\bibfield  {journal} {\bibinfo  {journal} {J. Mod. Opt.}\ }\textbf {\bibinfo {volume} {44}},\ \bibinfo {pages} {2455} (\bibinfo {year} {1997})}\BibitemShut {NoStop}%
\bibitem [{\citenamefont {Bongioanni}\ \emph {et~al.}(2010)\citenamefont {Bongioanni}, \citenamefont {Sansoni}, \citenamefont {Sciarrino}, \citenamefont {Vallone},\ and\ \citenamefont {Mataloni}}]{BSS+10}%
  \BibitemOpen
  \bibfield  {author} {\bibinfo {author} {\bibfnamefont {I.}~\bibnamefont {Bongioanni}}, \bibinfo {author} {\bibfnamefont {L.}~\bibnamefont {Sansoni}}, \bibinfo {author} {\bibfnamefont {F.}~\bibnamefont {Sciarrino}}, \bibinfo {author} {\bibfnamefont {G.}~\bibnamefont {Vallone}},\ and\ \bibinfo {author} {\bibfnamefont {P.}~\bibnamefont {Mataloni}},\ }\bibfield  {title} {\bibinfo {title} {Experimental quantum process tomography of non-trace-preserving maps},\ }\href@noop {} {\bibfield  {journal} {\bibinfo  {journal} {Phys. Rev. A}\ }\textbf {\bibinfo {volume} {82}},\ \bibinfo {pages} {042307} (\bibinfo {year} {2010})}\BibitemShut {NoStop}%
\bibitem [{\citenamefont {Mostafazadeh}(2002{\natexlab{b}})}]{Mos02}%
  \BibitemOpen
  \bibfield  {author} {\bibinfo {author} {\bibfnamefont {A.}~\bibnamefont {Mostafazadeh}},\ }\bibfield  {title} {\bibinfo {title} {{Pseudo-Hermiticity versus PT-symmetry III: Equivalence of pseudo-Hermiticity and the presence of antilinear symmetries}},\ }\href@noop {} {\bibfield  {journal} {\bibinfo  {journal} {J. Math. Phys.}\ }\textbf {\bibinfo {volume} {43}},\ \bibinfo {pages} {3944} (\bibinfo {year} {2002}{\natexlab{b}})}\BibitemShut {NoStop}%
\bibitem [{Note3()}]{Note3}%
  \BibitemOpen
  \bibinfo {note} {The ratio of non-zero measurement outcomes to the total number of copies approaches the success probability by the law of large numbers~\cite {DKLM05}. Our setting assumes that multiple copies of the state $\sigma $ are provided to the agent by an external agent who has the knowledge of $\sigma $, and not prepared by the agent implementing the inner-product changing channel, say, by cloning.}\BibitemShut {Stop}%
\bibitem [{\citenamefont {Stinespring}(1955)}]{Sti95}%
  \BibitemOpen
  \bibfield  {author} {\bibinfo {author} {\bibfnamefont {W.~F.}\ \bibnamefont {Stinespring}},\ }\bibfield  {title} {\bibinfo {title} {Positive functions on {$C^*$}-algebras},\ }\href@noop {} {\bibfield  {journal} {\bibinfo  {journal} {Proc. Am. Math. Soc.}\ }\textbf {\bibinfo {volume} {6}},\ \bibinfo {pages} {211} (\bibinfo {year} {1955})}\BibitemShut {NoStop}%
\bibitem [{\citenamefont {Paulsen}(2003)}]{Pau03}%
  \BibitemOpen
  \bibfield  {author} {\bibinfo {author} {\bibfnamefont {V.}~\bibnamefont {Paulsen}},\ }\href@noop {} {\emph {\bibinfo {title} {Completely Bounded Maps and Operator Algebras}}},\ Cambridge Studies in Advanced Mathematics\ (\bibinfo  {publisher} {Cambridge University Press},\ \bibinfo {year} {2003})\BibitemShut {NoStop}%
\bibitem [{\citenamefont {Sugiyama}\ \emph {et~al.}(2011)\citenamefont {Sugiyama}, \citenamefont {Turner},\ and\ \citenamefont {Murao}}]{STM11}%
  \BibitemOpen
  \bibfield  {author} {\bibinfo {author} {\bibfnamefont {T.}~\bibnamefont {Sugiyama}}, \bibinfo {author} {\bibfnamefont {P.~S.}\ \bibnamefont {Turner}},\ and\ \bibinfo {author} {\bibfnamefont {M.}~\bibnamefont {Murao}},\ }\bibfield  {title} {\bibinfo {title} {{Error probability analysis in quantum tomography: A tool for evaluating experiments}},\ }\href@noop {} {\bibfield  {journal} {\bibinfo  {journal} {Phys. Rev. A}\ }\textbf {\bibinfo {volume} {83}},\ \bibinfo {pages} {012105} (\bibinfo {year} {2011})}\BibitemShut {NoStop}%
\bibitem [{\citenamefont {Steinberg}\ \emph {et~al.}(1993)\citenamefont {Steinberg}, \citenamefont {Kwiat},\ and\ \citenamefont {Chiao}}]{SKC93}%
  \BibitemOpen
  \bibfield  {author} {\bibinfo {author} {\bibfnamefont {A.~M.}\ \bibnamefont {Steinberg}}, \bibinfo {author} {\bibfnamefont {P.~G.}\ \bibnamefont {Kwiat}},\ and\ \bibinfo {author} {\bibfnamefont {R.~Y.}\ \bibnamefont {Chiao}},\ }\bibfield  {title} {\bibinfo {title} {Measurement of the single-photon tunneling time},\ }\href@noop {} {\bibfield  {journal} {\bibinfo  {journal} {Phys. Rev. Lett.}\ }\textbf {\bibinfo {volume} {71}},\ \bibinfo {pages} {708} (\bibinfo {year} {1993})}\BibitemShut {NoStop}%
\bibitem [{\citenamefont {Morvan}\ \emph {et~al.}()\citenamefont {Morvan}, \citenamefont {Ramasesh}, \citenamefont {Blok}, \citenamefont {Kreikebaum}, \citenamefont {O'Brien}, \citenamefont {Chen}, \citenamefont {Mitchell}, \citenamefont {Naik}, \citenamefont {Santiago},\ and\ \citenamefont {Siddiqi}}]{BRS+20}%
  \BibitemOpen
  \bibfield  {author} {\bibinfo {author} {\bibfnamefont {A.}~\bibnamefont {Morvan}}, \bibinfo {author} {\bibfnamefont {V.~V.}\ \bibnamefont {Ramasesh}}, \bibinfo {author} {\bibfnamefont {M.~S.}\ \bibnamefont {Blok}}, \bibinfo {author} {\bibfnamefont {J.~M.}\ \bibnamefont {Kreikebaum}}, \bibinfo {author} {\bibfnamefont {K.}~\bibnamefont {O'Brien}}, \bibinfo {author} {\bibfnamefont {L.}~\bibnamefont {Chen}}, \bibinfo {author} {\bibfnamefont {B.~K.}\ \bibnamefont {Mitchell}}, \bibinfo {author} {\bibfnamefont {R.~K.}\ \bibnamefont {Naik}}, \bibinfo {author} {\bibfnamefont {D.~I.}\ \bibnamefont {Santiago}},\ and\ \bibinfo {author} {\bibfnamefont {I.}~\bibnamefont {Siddiqi}},\ }\bibfield  {title} {\bibinfo {title} {Qutrit randomized benchmarking},\ }\href@noop {} {\bibinfo  {journal} {arXiv:2008.09134}\ }\BibitemShut {NoStop}%
\bibitem [{\citenamefont {Kononenko}\ \emph {et~al.}()\citenamefont {Kononenko}, \citenamefont {Yurtalan}, \citenamefont {Shi},\ and\ \citenamefont {Lupascu}}]{KYSL20}%
  \BibitemOpen
\bibfield  {journal} {  }\bibfield  {author} {\bibinfo {author} {\bibfnamefont {M.}~\bibnamefont {Kononenko}}, \bibinfo {author} {\bibfnamefont {M.~A.}\ \bibnamefont {Yurtalan}}, \bibinfo {author} {\bibfnamefont {J.}~\bibnamefont {Shi}},\ and\ \bibinfo {author} {\bibfnamefont {A.}~\bibnamefont {Lupascu}},\ }\bibfield  {title} {\bibinfo {title} {Characterization of control in a superconducting qutrit using randomized benchmarking},\ }\href@noop {} {\bibinfo  {journal} {arXiv:2009.00599}\ }\BibitemShut {NoStop}%
\bibitem [{\citenamefont {Mostafazadeh}(2007)}]{Mos07}%
  \BibitemOpen
\bibfield  {journal} {  }\bibfield  {author} {\bibinfo {author} {\bibfnamefont {A.}~\bibnamefont {Mostafazadeh}},\ }\bibfield  {title} {\bibinfo {title} {Quantum brachistochrone problem and the geometry of the state space in pseudo-{Hermitian} quantum mechanics},\ }\href@noop {} {\bibfield  {journal} {\bibinfo  {journal} {Phys. Rev. Lett.}\ }\textbf {\bibinfo {volume} {99}},\ \bibinfo {pages} {130502} (\bibinfo {year} {2007})}\BibitemShut {NoStop}%
\bibitem [{\citenamefont {Graefe}\ \emph {et~al.}(2010)\citenamefont {Graefe}, \citenamefont {Korsch},\ and\ \citenamefont {Niederle}}]{GKN10}%
  \BibitemOpen
  \bibfield  {author} {\bibinfo {author} {\bibfnamefont {E.-M.}\ \bibnamefont {Graefe}}, \bibinfo {author} {\bibfnamefont {H.~J.}\ \bibnamefont {Korsch}},\ and\ \bibinfo {author} {\bibfnamefont {A.~E.}\ \bibnamefont {Niederle}},\ }\bibfield  {title} {\bibinfo {title} {Quantum-classical correspondence for a non-{Hermitian Bose-Hubbard} dimer},\ }\href@noop {} {\bibfield  {journal} {\bibinfo  {journal} {Phys. Rev. A}\ }\textbf {\bibinfo {volume} {82}},\ \bibinfo {pages} {013629} (\bibinfo {year} {2010})}\BibitemShut {NoStop}%
\bibitem [{\citenamefont {Uzdin}\ \emph {et~al.}(2012)\citenamefont {Uzdin}, \citenamefont {G{\"u}nther}, \citenamefont {Rahav},\ and\ \citenamefont {Moiseyev}}]{UGRM12}%
  \BibitemOpen
  \bibfield  {author} {\bibinfo {author} {\bibfnamefont {R.}~\bibnamefont {Uzdin}}, \bibinfo {author} {\bibfnamefont {U.}~\bibnamefont {G{\"u}nther}}, \bibinfo {author} {\bibfnamefont {S.}~\bibnamefont {Rahav}},\ and\ \bibinfo {author} {\bibfnamefont {N.}~\bibnamefont {Moiseyev}},\ }\bibfield  {title} {\bibinfo {title} {Time-dependent {Hamiltonians} with 100\% evolution speed efficiency},\ }\href@noop {} {\bibfield  {journal} {\bibinfo  {journal} {J. Phys. A: Math. Theor.}\ }\textbf {\bibinfo {volume} {45}},\ \bibinfo {pages} {415304} (\bibinfo {year} {2012})}\BibitemShut {NoStop}%
\bibitem [{\citenamefont {Dekking}\ \emph {et~al.}(2005)\citenamefont {Dekking}, \citenamefont {Kraaikamp}, \citenamefont {Lopuha\"{a}},\ and\ \citenamefont {Meester}}]{DKLM05}%
  \BibitemOpen
  \bibfield  {author} {\bibinfo {author} {\bibfnamefont {F.~M.}\ \bibnamefont {Dekking}}, \bibinfo {author} {\bibfnamefont {C.}~\bibnamefont {Kraaikamp}}, \bibinfo {author} {\bibfnamefont {H.~P.}\ \bibnamefont {Lopuha\"{a}}},\ and\ \bibinfo {author} {\bibfnamefont {L.~E.}\ \bibnamefont {Meester}},\ }\href@noop {} {\emph {\bibinfo {title} {{A Modern Introduction to Probability and Statistics: Understanding Why and How}}}}\ (\bibinfo  {publisher} {Springer-Verlag},\ \bibinfo {address} {London},\ \bibinfo {year} {2005})\BibitemShut {NoStop}%
\end{thebibliography}%
\end{document}